\documentclass[11pt]{article}

\usepackage[pagebackref,colorlinks]{hyperref}
\usepackage{amsmath} % AMS Math Package
\usepackage{amsthm} % Theorem Formatting
\usepackage{amssymb}	% Math symbols such as \mathbb
\usepackage{graphicx} % Allows for eps images
\usepackage{multicol} % Allows for multiple columns
\usepackage{multirow}
\usepackage{color}
\usepackage[dvips,letterpaper,margin=1in,bottom=1in]{geometry}
\usepackage[capitalize,noabbrev]{cleveref}

\usepackage[utf8]{inputenc}
\usepackage[english]{babel}

\usepackage[T1]{fontenc}
\AtBeginDocument{%
  \DeclareFontShape{T1}{cmr}{m}{scit}{<->ssub*cmr/m/sc}{}%
}

\usepackage{diagbox}
\usepackage{mathtools}

\newtheorem{theorem}{Theorem}[section]

\newtheorem{lemma}[theorem]{Lemma}
\newtheorem{corollary}[theorem]{Corollary}

\newtheorem{fact}[theorem]{Fact}

\newtheorem{definition}[theorem]{Definition}

\newtheorem{remark}{Remark}[section]

\DeclarePairedDelimiter\rbra{\lparen}{\rparen}
\DeclarePairedDelimiter\sbra{\lbrack}{\rbrack}
\DeclarePairedDelimiter\cbra{\{}{\}}
\DeclarePairedDelimiter\abs{\lvert}{\rvert}
\DeclarePairedDelimiter\Abs{\lVert}{\rVert}

\DeclarePairedDelimiter\floor{\lfloor}{\rfloor}
\DeclarePairedDelimiter\ket{\lvert}{\rangle}
\DeclarePairedDelimiter\bra{\langle}{\rvert}
\DeclarePairedDelimiter\ave{\langle}{\rangle}

\newcommand{\set}[2] {\left\{\, #1 \colon #2 \,\right\}}

\newcommand{\tr} {\operatorname{tr}}

\newcommand{\polylog} {\operatorname{polylog}}

\newcommand{\Sym}{\operatorname{Sym}}

\newcommand{\ketbra}[2]{\ensuremath{\ket{#1}\!\bra{#2}}}

\usepackage{latexsym}
\usepackage{CJK}

\usepackage{enumerate}

\usepackage{algorithm}
\usepackage{algpseudocode}
\usepackage{stmaryrd}
\usepackage{tabularx}
\usepackage{booktabs}
\usepackage{adjustbox} % Adjust the maximum width of tables

\newcommand{\footremember}[2]{%
    \footnote{#2}
    \newcounter{#1}
    \setcounter{#1}{\value{footnote}}%
}

\usepackage{tikz}
\usetikzlibrary{quantikz2}

\begin{document}
    \title{Quantum Time Lower Bounds by Permutation Invariance}
\author{
    Qisheng Wang \footremember{1}{Qisheng Wang is with the School of Computer Science, Shanghai Jiao Tong University, Shanghai, China (e-mail: \href{mailto:QishengWang1994@gmail.com}{\nolinkurl{QishengWang1994@gmail.com}}).}
}
        \date{}
        \maketitle

    \begin{abstract}
        Tight bounds on quantum sample complexity and quantum query complexity have been known for various computational problems in the literature, whereas tight bounds on quantum time complexity (i.e., the size of quantum circuits) remain unresolved. 
        In this paper, we provide a framework to establish lower bounds on the quantum time complexity for testing permutation-invariant properties of quantum states, via a reduction from quantum sample complexity. 
        As an application, we obtain a series of \textit{matching} lower bounds when given sample access to the input quantum states, including:
        \begin{enumerate}
            \item The SWAP test due to \hyperlink{cite.BCWdW01}{Buhrman, Cleve, Watrous, and de Wolf (\textit{Phys.\ Rev.\ Lett.\ }2001)} is time-optimal to estimate the purity $\operatorname{tr}(\rho^2)$ and the inner product $\operatorname{tr}(\rho\sigma)$. 
            \item The Shift test due to \hyperlink{cite.EAO+02}{Ekert, Alves, Oi, Horodecki, Horodecki, and Kwek (\textit{Phys.\ Rev.\ Lett.\ }2002)} is time-optimal to estimate the high-order functionals $\operatorname{tr}(\rho^k)$. 
            \item The productness tester for multipartite pure states due to \hyperlink{cite.HM13}{Harrow and Montanaro (\textit{J.\ ACM} 2013)} is time-optimal. 
            \item The LMR protocol due to \hyperlink{cite.LMR14}{Lloyd, Mohseni, and Rebentrost (\textit{Nat.\ Phys.\ }2014)} is time-optimal to implement the reflection operator about a pure state.
            \item The samplizer due to \hyperlink{cite.WZ24}{Wang and Zhang (\textit{IEEE Trans.\ Inf.\ Theory} 2025)} is time-optimal for pure states. 
            \item The estimator for pure-state trace distance and fidelity due to \hyperlink{cite.WZ24c}{Wang and Zhang (ICALP 2026)} is time-optimal. 
        \end{enumerate}
        To the best of our knowledge, this is the \textit{first} method that allows us to systematically establish tight lower bounds on quantum time complexity. 
    \end{abstract}


    \newpage
    \tableofcontents
    \newpage

    \section{Introduction}

    Quantum time complexity measures the efficiency of a quantum algorithm running on a quantum computer, typically defined by the number of elementary quantum gates used during the execution of the algorithm. 
    The tasks that can be efficiently solved by quantum computing are characterized by the computational complexity class $\mathsf{BQP}$ \cite{BV97}, the set of decision problems that can be solved by a polynomial-time quantum algorithm. 
    Optimizing the time complexity of quantum algorithms is therefore at the core of quantum computing.
    During the development of finding time-efficient quantum algorithms, several tools have been proposed, e.g., history-independent data structures \cite{Amb07}, time complexity version of quantum subroutine composition \cite{Jef22,BJY24}, and time-efficient implementations of quantum walks \cite{BCJ+13,JZ23}, span programs \cite{CJOP20}, and quantum divide and conquer \cite{ABB+23}. 
    Time-efficient tools have also been found for quantum property testing \cite{MdW16}, e.g., 
    quantum Schur transform \cite{BCH06,BCH07,Ngu23,GBO23}, weak Schur sampling \cite{CHW07}, and density matrix exponentiation \cite{LMR14,KLL+17,GKP+24}, enabling applications such as quantum state tomography \cite{HHJ+17,OW16,OW17}, quantum state certification \cite{OW21,BOW19,WZ24b}, and quantum entropy estimation \cite{AISW20,BMW16,WZ24}. 

    Previous techniques for quantum lower bounds focused mainly on quantum communication complexity (cf.\ \cite{Bra03}), quantum query complexity \cite{BBC+01,Amb02} and quantum sample complexity (e.g., \cite{CHW07,OW21}). 
    However, methodologies for proving lower bounds on quantum time complexity remain open.
    As evidence, quantum algorithms with optimal query/sample complexity do not necessarily achieve optimal time complexity. 
    For example, given that the quantum query complexity for unstructured search on $N$ items is known to be $\Theta\rbra{\sqrt{N}}$ \cite{BBBV97,Gro96,BBHT98,Zal99}, its quantum time complexity is then trivially $\Omega\rbra{\sqrt{N}}$, which is, however, not known to be tight as the textbook time complexity $O\rbra{\sqrt{N}\log\rbra{N}}$ \cite{NC10} was later improved to $O\rbra{\sqrt{N}\log\rbra{\log\rbra{N}}}$ in \cite{Gro02} and further to $O\rbra{\sqrt{N}\log\rbra{\log^*\rbra{N}}}$ in \cite{AdW17}.\footnote{Here, $\log^*\rbra{N}$ is the iterated logarithm of $N$, defined by $\log^*\rbra{N} = 1 + \log^*\rbra{\log\rbra{N}}$ for $N > 1$ and $0$ otherwise.}
    Another example is the quantum state certification with respect to trace distance, whose sample complexity is known to be $\Theta\rbra{N/\varepsilon^2}$ \cite{BOW19} for $N$-dimensional quantum states and precision $\varepsilon$ but with time complexity $O\rbra{N^3/\varepsilon^6}$, whereas a different approach in \cite{WZ24b} achieves a better time complexity of $\widetilde{O}\rbra{N^2/\varepsilon^5}$ (but with worse sample complexity).\footnote{$\widetilde{O}\rbra{\cdot}$ suppresses polylogarithmic factors.} 
    It can be seen that quantum time complexity can be even more challenging to characterize than quantum query/sample complexity. 
    This naturally raises the following question:
    \[
    \textit{Is there any method for systematically proving tight quantum \textbf{time} lower bounds?}
    \]

    In this paper, we propose a method for proving lower bounds on the \textit{quantum time complexity} for the quantum property testing under the permutation-invariant condition. 
    As an application, we show the time-optimality of a series of quantum algorithmic tools, including the SWAP test \cite{BCWdW01}, Shift test \cite{EAO+02}, productness tester \cite{HM13}, LMR protocol \cite{LMR14,KLL+17,GKP+24}, and samplizer \cite{WZ24,WZ24c}.
    
    To the best of our knowledge, this is the \textit{first} method that allows us to systematically derive \textit{optimal} lower bounds on quantum time complexity \textit{up to a constant factor}. 

    \subsection{Main Results} \label{sec:main}

    A property of $n$-qubit (mixed) quantum states (hereinafter referred to as ``$n$-qubit property'') is denoted by a pair of disjoint sets $\mathcal{P}_n = \rbra{\mathcal{P}_n^{\textup{yes}}, \mathcal{P}_n^{\textup{no}}}$, each consisting of $n$-qubit states. 
    A tester for $\mathcal{P}_n$ can determine whether $\rho \in \mathcal{P}_n^{\textup{yes}}$ or $\rho \in \mathcal{P}_n^{\textup{no}}$ (promised that it is in either case) with probability at least $2/3$ from the quantum state $\rho^{\otimes S}$ (independent and identical samples of $\rho$), where the sample complexity is $S$, the number of samples of $\rho$, and the time complexity is the number of elementary quantum gates (and measurements). 
    The sample/time complexity of $\mathcal{P}_n$, denoted by $\mathsf{S}\rbra{\mathcal{P}_n}$/$\mathsf{T}\rbra{\mathcal{P}_n}$, is the minimum sample/time complexity over all (non-uniform) testers of $\mathcal{P}_n$ (see \cref{sec:tester} for the formal definition). 
    For example, for $N$-dimensional quantum state certification $\textsc{QSD}\sbra{N, \varepsilon}$ to precision $\varepsilon$ with respect to trace distance, it is known that $\mathsf{S}\rbra{\textsc{QSD}\sbra{N, \varepsilon}} = \Theta\rbra{N/\varepsilon^2}$ \cite{BOW19} and $\mathsf{T}\rbra{\textsc{QSD}\sbra{N, \varepsilon}} = \widetilde{O}\rbra{N^2/\varepsilon^5}$ \cite{WZ24b}. 
    In general, it trivially holds that $\mathsf{T}\rbra{\mathcal{P}_n} \geq \mathsf{S}\rbra{\mathcal{P}_n}$, which is already tight.\footnote{\label{fn:tight-example}A simple explanation for this is that any sample should relate to at least one gate or measurement (otherwise this sample is not necessary). On the contrary, this is trivially tight when, for example, testing whether the first qubit of an $n$-qubit quantum state is $\ket{0}$ or $\ket{1}$.} 
    A direct question is: can we improve this relation so that we can establish matching lower bounds on the quantum time complexity in certain cases of interest?

    In this paper, we answer this question for the permutation-invariant case with embeddability. 
    Permutation invariance is a basic symmetry in physics \cite{FR09} and is known to have interesting applications in quantum information theory \cite{Wat18}. 
    An $n$-qubit property is said to be permutation-invariant, if it is invariant under qubit-permutation action for every permutation $\pi \in \Sym\rbra{n}$ (see \cref{def:perm-inv} for the formal definition). 
    A property $\mathcal{Q}$ is said to be embeddable in another property $\mathcal{P}$, if there is an embedding state $\sigma$ such that $\rho$ satisfies $\mathcal{Q}$ if and only if $\rho \otimes \sigma$ satisfies $\mathcal{P}$ for every state $\rho$ (see \cref{def:embed} for the formal definition).
    
    We first present the simplest yet useful case of our results. 

    \begin{theorem} [Sample-to-time reduction for permutation-invariant properties, \cref{thm:main-sp}] \label{thm:pi-intro}
        If a $1$-qubit property $\mathcal{Q}_1$ is embeddable in an $n$-qubit permutation-invariant property $\mathcal{P}_n$, then 
        \begin{equation}
            \mathsf{T}\rbra{\mathcal{P}_n} \geq n \cdot \mathsf{S}\rbra{\mathcal{Q}_1}.
        \end{equation}
    \end{theorem}

    Purity estimation, for example, is a simple application of \cref{thm:pi-intro}, as the purity $\tr\rbra{\rho^2}$ is undoubtedly permutation-invariant.
    As will be shown in \cref{sec:app}, \cref{thm:pi-intro} can be used to establish tight quantum time lower bounds for purity estimation and a series of other quantum property testing problems. 
    These lower bounds yield the time-optimality of several useful quantum algorithmic tools such as the SWAP test \cite{BCWdW01}, Shift test (generalized SWAP test) \cite{EAO+02}, LMR protocol \cite{LMR14}, and samplizer \cite{WZ24}. 

    To make it more powerful, we generalize \cref{thm:pi-intro} to the case where the invariance only holds for certain permutation groups and the embeddability holds for multiple qubits. 

    \begin{theorem} [Sample-to-time reduction for partially permutation-invariant embeddability, \cref{thm:main}] \label{thm:full-intro}
        Let $\mathcal{G}$ be a permutation subgroup of the form $\mathcal{G} = \Sym\rbra{A_1} \times \Sym\rbra{A_2} \times \dots \times \Sym\rbra{A_k}$, where $A_1, A_2, \dots, A_k$ form a partition of $\sbra{n}$.\footnote{We write $\sbra{n} = \cbra{1, 2, \dots, n}$ and use $\Sym\rbra{A}$ to denote the symmetric group over the set $A$.} 
        If an $m$-qubit property $\mathcal{Q}_m$ is $\mathcal{G}$-invariantly embeddable in an $n$-qubit property $\mathcal{P}_n$, then 
        \begin{equation}
            \mathsf{T}\rbra{\mathcal{P}_n} \geq R \cdot \mathsf{S}\rbra{\mathcal{Q}_m}, \text{ where } R = \min_{j \in \sbra{k} \colon A_j \cap \sbra{m} \neq \emptyset} \floor*{\frac{\abs{A_j}}{\abs*{A_j \cap \sbra{m}}}}.
        \end{equation}
        Here, the $\mathcal{G}$-invariant embeddability (see \cref{def:inv-embed}) means the existence of a state $\sigma$ such that for every $\rho \in \mathcal{Q}_m^X$ with $X \in \cbra{\textup{yes}, \textup{no}}$, $U_{\pi} \rbra{\rho \otimes \sigma} U_{\pi}^\dag \in \mathcal{P}_n^X$ for every permutation $\pi \in \mathcal{G}$, where $U_\pi$ means the qubit-permutation operator for $\pi$.
    \end{theorem}

    The parameter $R$ in \cref{thm:full-intro} depends only on the embeddability of $\mathcal{Q}_m$ into $\mathcal{P}_n$ (specifically, $\mathcal{G}$ and $m$) but not the property $\mathcal{Q}_m$ itself. 
    \cref{thm:pi-intro} is actually a special case of \cref{thm:full-intro} with $m = k = 1$ and $A_1 = \sbra{n}$ (in which case $R = n$). 
    In comparison, \cref{thm:full-intro} is applicable to more general cases where qubits are distinguishable to some extent. 
    As an application, we show that the multipartite productness tester in \cite{HM13} is time-optimal. 
    See \cref{sec:app} for more discussions. 

    \begin{remark} [Extensibility to qudits] \label{remark:qudit}
        Our results can be naturally extended to the case where quantum states are made of qudits (and thus an elementary quantum gate means a $2$-qudit gate, see \cref{thm:main-qudit,thm:main-sp-qudit}). 
        Here, the (pure) quantum state of a $d$-dimensional qudit is described by a linear combination of $\ket{0}, \ket{1}, \dots, \ket{d-1}$. In particular, a qubit is a $2$-dimensional qudit.
        This extended version is useful, for example, for proving \cref{corollary:inner-product} using $d = 4$. 
    \end{remark}
    
    \subsection{Applications} \label{sec:app}

    As an application, we prove a series of tight quantum time lower bounds, with which we show the time-optimality of several quantum algorithmic tools proposed in the literature. 
    The relationships amongst them are presented in \cref{fig:app}. 

\begin{figure}[t]
    \centering
\adjustbox{max width=\textwidth}{
\tikzset{every picture/.style={line width=0.75pt}} %set default line width to 0.75pt        

\begin{tikzpicture}[x=0.75pt,y=0.75pt,yscale=-1,xscale=1]
%uncomment if require: \path (0,783); %set diagram left start at 0, and has height of 783

%Shape: Rectangle [id:dp787812591532582] 
\draw   (10,330) -- (170,330) -- (170,390) -- (10,390) -- cycle ;
%Shape: Rectangle [id:dp2075518234868089] 
\draw   (10,130) -- (170,130) -- (170,190) -- (10,190) -- cycle ;
%Shape: Rectangle [id:dp07503050754212037] 
\draw   (250,230) -- (410,230) -- (410,270) -- (250,270) -- cycle ;
%Shape: Rectangle [id:dp26824471064898214] 
\draw   (250,300) -- (410,300) -- (410,340) -- (250,340) -- cycle ;
%Shape: Rectangle [id:dp22040824507343426] 
\draw   (250,370) -- (410,370) -- (410,410) -- (250,410) -- cycle ;
%Shape: Rectangle [id:dp579244916794269] 
\draw   (250,440) -- (410,440) -- (410,480) -- (250,480) -- cycle ;
%Shape: Rectangle [id:dp3348792501967376] 
\draw   (250,140) -- (410,140) -- (410,180) -- (250,180) -- cycle ;
%Shape: Rectangle [id:dp9387822568152019] 
\draw   (490,265) -- (650,265) -- (650,305) -- (490,305) -- cycle ;
%Shape: Rectangle [id:dp5271926033059137] 
\draw   (490,370) -- (650,370) -- (650,410) -- (490,410) -- cycle ;
%Shape: Rectangle [id:dp5889845415284048] 
\draw   (490,510) -- (650,510) -- (650,550) -- (490,550) -- cycle ;
%Shape: Rectangle [id:dp25766819967564065] 
\draw   (490,580) -- (650,580) -- (650,620) -- (490,620) -- cycle ;
%Shape: Rectangle [id:dp21878915718568548] 
\draw   (490,440) -- (650,440) -- (650,480) -- (490,480) -- cycle ;
%Shape: Rectangle [id:dp7525841591808132] 
\draw   (490,140) -- (650,140) -- (650,180) -- (490,180) -- cycle ;
%Straight Lines [id:da9035382081559761] 
\draw    (170,360) .. controls (169.63,357.67) and (170.61,356.33) .. (172.94,355.96) .. controls (175.27,355.59) and (176.25,354.24) .. (175.88,351.91) .. controls (175.51,349.58) and (176.49,348.24) .. (178.82,347.87) .. controls (181.15,347.5) and (182.13,346.16) .. (181.76,343.83) .. controls (181.39,341.5) and (182.37,340.15) .. (184.7,339.78) .. controls (187.03,339.41) and (188.01,338.07) .. (187.65,335.74) .. controls (187.28,333.41) and (188.26,332.06) .. (190.59,331.69) .. controls (192.92,331.32) and (193.9,329.98) .. (193.53,327.65) .. controls (193.16,325.32) and (194.14,323.98) .. (196.47,323.61) .. controls (198.8,323.24) and (199.78,321.89) .. (199.41,319.56) .. controls (199.04,317.23) and (200.02,315.89) .. (202.35,315.52) .. controls (204.68,315.15) and (205.66,313.81) .. (205.29,311.48) .. controls (204.92,309.15) and (205.9,307.8) .. (208.23,307.43) .. controls (210.56,307.06) and (211.54,305.72) .. (211.17,303.39) .. controls (210.8,301.06) and (211.78,299.71) .. (214.11,299.34) .. controls (216.44,298.97) and (217.42,297.63) .. (217.05,295.3) .. controls (216.68,292.97) and (217.66,291.63) .. (219.99,291.26) .. controls (222.32,290.89) and (223.31,289.54) .. (222.94,287.21) .. controls (222.57,284.88) and (223.55,283.54) .. (225.88,283.17) .. controls (228.21,282.8) and (229.19,281.46) .. (228.82,279.13) .. controls (228.45,276.8) and (229.43,275.45) .. (231.76,275.08) .. controls (234.09,274.71) and (235.07,273.37) .. (234.7,271.04) .. controls (234.33,268.71) and (235.31,267.37) .. (237.64,267) .. controls (239.97,266.63) and (240.95,265.28) .. (240.58,262.95) .. controls (240.21,260.62) and (241.19,259.28) .. (243.52,258.91) -- (244.12,258.09) -- (248.82,251.62) ;
\draw [shift={(250,250)}, rotate = 126.03] [color={rgb, 255:red, 0; green, 0; blue, 0 }  ][line width=0.75]    (10.93,-3.29) .. controls (6.95,-1.4) and (3.31,-0.3) .. (0,0) .. controls (3.31,0.3) and (6.95,1.4) .. (10.93,3.29)   ;
%Straight Lines [id:da7570641835740158] 
\draw    (170,360) -- (248.21,320.89) ;
\draw [shift={(250,320)}, rotate = 153.43] [color={rgb, 255:red, 0; green, 0; blue, 0 }  ][line width=0.75]    (10.93,-3.29) .. controls (6.95,-1.4) and (3.31,-0.3) .. (0,0) .. controls (3.31,0.3) and (6.95,1.4) .. (10.93,3.29)   ;
%Straight Lines [id:da14477393222960577] 
\draw    (170,360) -- (248.13,389.3) ;
\draw [shift={(250,390)}, rotate = 200.56] [color={rgb, 255:red, 0; green, 0; blue, 0 }  ][line width=0.75]    (10.93,-3.29) .. controls (6.95,-1.4) and (3.31,-0.3) .. (0,0) .. controls (3.31,0.3) and (6.95,1.4) .. (10.93,3.29)   ;
%Straight Lines [id:da7216272626133718] 
\draw    (170,360) -- (248.75,458.44) ;
\draw [shift={(250,460)}, rotate = 231.34] [color={rgb, 255:red, 0; green, 0; blue, 0 }  ][line width=0.75]    (10.93,-3.29) .. controls (6.95,-1.4) and (3.31,-0.3) .. (0,0) .. controls (3.31,0.3) and (6.95,1.4) .. (10.93,3.29)   ;
%Straight Lines [id:da06348656099716832] 
\draw    (170,160) -- (248,160) ;
\draw [shift={(250,160)}, rotate = 180] [color={rgb, 255:red, 0; green, 0; blue, 0 }  ][line width=0.75]    (10.93,-3.29) .. controls (6.95,-1.4) and (3.31,-0.3) .. (0,0) .. controls (3.31,0.3) and (6.95,1.4) .. (10.93,3.29)   ;
%Straight Lines [id:da7280020719882423] 
\draw    (410,250) -- (487.17,284.19) ;
\draw [shift={(489,285)}, rotate = 203.9] [color={rgb, 255:red, 0; green, 0; blue, 0 }  ][line width=0.75]    (10.93,-3.29) .. controls (6.95,-1.4) and (3.31,-0.3) .. (0,0) .. controls (3.31,0.3) and (6.95,1.4) .. (10.93,3.29)   ;
%Straight Lines [id:da8964922196815932] 
\draw    (410,320) -- (487.17,285.81) ;
\draw [shift={(489,285)}, rotate = 156.1] [color={rgb, 255:red, 0; green, 0; blue, 0 }  ][line width=0.75]    (10.93,-3.29) .. controls (6.95,-1.4) and (3.31,-0.3) .. (0,0) .. controls (3.31,0.3) and (6.95,1.4) .. (10.93,3.29)   ;
%Straight Lines [id:da8773335074051511] 
\draw    (410,390) -- (488,390) ;
\draw [shift={(490,390)}, rotate = 180] [color={rgb, 255:red, 0; green, 0; blue, 0 }  ][line width=0.75]    (10.93,-3.29) .. controls (6.95,-1.4) and (3.31,-0.3) .. (0,0) .. controls (3.31,0.3) and (6.95,1.4) .. (10.93,3.29)   ;
%Straight Lines [id:da14363822315931118] 
\draw    (410,460) -- (488,460) ;
\draw [shift={(490,460)}, rotate = 180] [color={rgb, 255:red, 0; green, 0; blue, 0 }  ][line width=0.75]    (10.93,-3.29) .. controls (6.95,-1.4) and (3.31,-0.3) .. (0,0) .. controls (3.31,0.3) and (6.95,1.4) .. (10.93,3.29)   ;
%Straight Lines [id:da03523917381221342] 
\draw    (410,160) -- (488,160) ;
\draw [shift={(490,160)}, rotate = 180] [color={rgb, 255:red, 0; green, 0; blue, 0 }  ][line width=0.75]    (10.93,-3.29) .. controls (6.95,-1.4) and (3.31,-0.3) .. (0,0) .. controls (3.31,0.3) and (6.95,1.4) .. (10.93,3.29)   ;
%Straight Lines [id:da9564108308782985] 
\draw    (570,180) -- (570,261) ;
\draw [shift={(570,263)}, rotate = 270] [color={rgb, 255:red, 0; green, 0; blue, 0 }  ][line width=0.75]    (10.93,-3.29) .. controls (6.95,-1.4) and (3.31,-0.3) .. (0,0) .. controls (3.31,0.3) and (6.95,1.4) .. (10.93,3.29)   ;
%Shape: Rectangle [id:dp8175082678116469] 
\draw  [dash pattern={on 0.84pt off 2.51pt}] (0,40) -- (180,40) -- (180,630) -- (0,630) -- cycle ;
%Shape: Rectangle [id:dp6724392307721321] 
\draw  [dash pattern={on 0.84pt off 2.51pt}] (240,40) -- (420,40) -- (420,630) -- (240,630) -- cycle ;
%Shape: Rectangle [id:dp1718190725319897] 
\draw  [dash pattern={on 0.84pt off 2.51pt}] (480,40) -- (660,40) -- (660,630) -- (480,630) -- cycle ;
%Straight Lines [id:da30758940931366063] 
\draw    (570,480) -- (570,508) ;
\draw [shift={(570,510)}, rotate = 270] [color={rgb, 255:red, 0; green, 0; blue, 0 }  ][line width=0.75]    (10.93,-3.29) .. controls (6.95,-1.4) and (3.31,-0.3) .. (0,0) .. controls (3.31,0.3) and (6.95,1.4) .. (10.93,3.29)   ;
%Straight Lines [id:da9395794373679177] 
\draw    (570,550) -- (570,578) ;
\draw [shift={(570,580)}, rotate = 270] [color={rgb, 255:red, 0; green, 0; blue, 0 }  ][line width=0.75]    (10.93,-3.29) .. controls (6.95,-1.4) and (3.31,-0.3) .. (0,0) .. controls (3.31,0.3) and (6.95,1.4) .. (10.93,3.29)   ;
%Striped Right Arrow [id:dp23267521688434845] 
\draw   (77.5,322.5) -- (77.5,246) -- (70,246) -- (85,190) -- (100,246) -- (92.5,246) -- (92.5,322.5) -- cycle ;\draw   (77.5,330) -- (77.5,328.5) -- (92.5,328.5) -- (92.5,330) -- cycle ;\draw   (77.5,327) -- (77.5,324) -- (92.5,324) -- (92.5,327) -- cycle ;

% Text Node
\draw (12,333) node [anchor=north west][inner sep=0.75pt]   [align=left] {{\small Sample-to-time reduction}\\{\footnotesize by permutation invariance}\\{\small \cref{thm:pi-intro}}};
% Text Node
\draw (12,133) node [anchor=north west][inner sep=0.75pt]   [align=left] {{\small Sample-to-time reduction}\\{\footnotesize by certain group invariance}\\{\small \cref{thm:full-intro}}};
% Text Node
\draw (94.5,247) node [anchor=north west][inner sep=0.75pt]   [align=left] {generalize};
% Text Node
\draw (252,233) node [anchor=north west][inner sep=0.75pt]   [align=left] {{\small Inner product estimation}\\{\small \cref{corollary:inner-product}}};
% Text Node
\draw (252,303) node [anchor=north west][inner sep=0.75pt]   [align=left] {{\small Purity estimation/testing}\\{\small \cref{corollary:purity-intro}}};
% Text Node
\draw (252,373) node [anchor=north west][inner sep=0.75pt]   [align=left] {{\small Power trace estimation}\\{\small \cref{corollary:pow-tr-intro}}};
% Text Node
\draw (252,443) node [anchor=north west][inner sep=0.75pt]   [align=left] {{\small Trace distance estimation}\\{\small \cref{corollary:td-estimation}}};
% Text Node
\draw (252,143) node [anchor=north west][inner sep=0.75pt]   [align=left] {{\small Productness testing}\\{\small \cref{corollary:prod-test}}};
% Text Node
\draw (492,275) node [anchor=north west][inner sep=0.75pt]   [align=left] {SWAP test \cite{BCWdW01}};
% Text Node
\draw (492,380) node [anchor=north west][inner sep=0.75pt]   [align=left] {Shift test \cite{EAO+02}};
% Text Node
\draw (492,513) node [anchor=north west][inner sep=0.75pt]   [align=left] {Samplizer \cite{WZ24} \\ \cref{corollary:samplizer}};
% Text Node
\draw (492,583) node [anchor=north west][inner sep=0.75pt]   [align=left] {LMR protocol \cite{LMR14} \\ \cref{corollary:lmr}};
% Text Node
\draw (492,443) node [anchor=north west][inner sep=0.75pt]   [align=left] {Estimator for pure-state \\ trace distance \cite{WZ24c}};
% Text Node
\draw (492,150) node [anchor=north west][inner sep=0.75pt]   [align=left] {Product test \cite{HM13}};
% Text Node
\draw (267,81) node [anchor=north west][inner sep=0.75pt]   [align=left] {Time Lower Bounds};
% Text Node
\draw (518,81) node [anchor=north west][inner sep=0.75pt]   [align=left] {Time-Optimality};
% Text Node
\draw (59,81) node [anchor=north west][inner sep=0.75pt]   [align=left] {Methods};
% Text Node
\draw (176.62,314.88) node [anchor=north west][inner sep=0.75pt]  [rotate=-306.32] [align=left] {\small qudit version};

\end{tikzpicture}
}
    \caption{Diagram of relationships amongst our results.}
    \label{fig:app}
\end{figure}
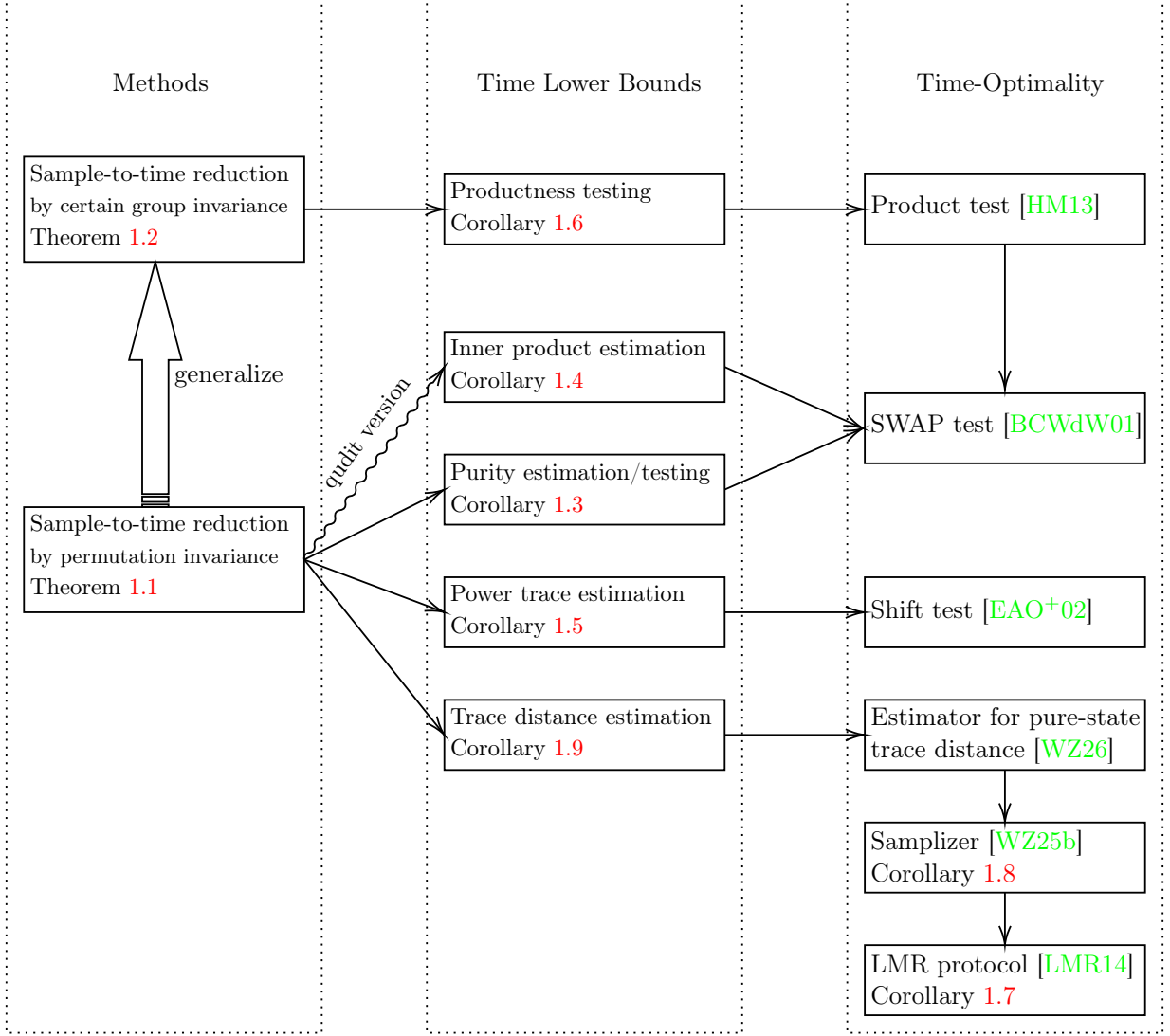

    In the following, we introduce the time-optimality of each quantum algorithmic tool. 

    \paragraph{Time-optimality of SWAP test.}

    The purity $\tr\rbra{\rho^2}$ is unitary-invariant (and thus permutation-invariant, as mentioned in \cref{sec:main}), which can be estimated by the SWAP test \cite{BCWdW01}. 
    By \cref{thm:pi-intro}, we can establish tight quantum time complexity lower bounds for purity estimation and purity testing, implying that the SWAP test is time-optimal for the two tasks. 

    \begin{corollary} [Quantum time lower bounds for purity estimation/testing, \cref{thm:lb-purity}] \label{corollary:purity-intro}
        Given sample access to an $n$-qubit quantum state $\rho$, (i) estimating $\tr\rbra{\rho^2}$ to within additive error $\varepsilon$ requires quantum time complexity $\Omega\rbra{n/\varepsilon^2}$, and (ii) determining whether $\tr\rbra{\rho^2} = 1$ or $\tr\rbra{\rho^2} \leq 1 - \varepsilon$ requires quantum time complexity $\Omega\rbra{n/\varepsilon}$.
    \end{corollary}
    
    Note that by the SWAP test \cite{BCWdW01}, purity estimation can be solved with sample complexity $O\rbra{1/\varepsilon^2}$ and time complexity $O\rbra{n/\varepsilon^2}$, and purity testing can be solved with sample complexity $O\rbra{1/\varepsilon}$ and time complexity $O\rbra{n/\varepsilon}$. 
    
    \begin{proof} [Proof sketch of \cref{corollary:purity-intro}]
        This can be shown by \cref{thm:pi-intro} and noting the sample complexity lower bounds $\Omega\rbra{1/\varepsilon^2}$ for purity estimation \cite{CWLY23,GHYZ24} and $\Omega\rbra{1/\varepsilon}$ for purity testing \cite{SW22,CWZ24}. 
    \end{proof}

    In addition, our method also applies to non-unitary-invariant (but still permutation-invariant) cases, e.g., inner product estimation.
    The inner product $\tr\rbra{\rho\sigma}$, which is the (squared) fidelity when one of $\rho$ and $\sigma$ is pure, can be estimated by the SWAP test \cite{BCWdW01}. 
    By the qudit version of \cref{thm:pi-intro}, we can establish tight quantum time complexity lower bounds for inner product estimation, implying that the SWAP test is time-optimal for this task. 

    \begin{corollary} [Quantum time lower bounds for inner product estimation, \cref{thm:lb-inner-prod}] \label{corollary:inner-product}
        Given sample access to $n$-qubit quantum states $\rho$ and $\sigma$, estimating $\tr\rbra{\rho\sigma}$ to within additive error $\varepsilon$ requires quantum time complexity $\Omega\rbra{n/\varepsilon^2}$, even if both $\rho$ and $\sigma$ are pure states. 
    \end{corollary}

    Note that by the SWAP test \cite{BCWdW01}, inner product estimation can be solved with sample complexity $O\rbra{1/\varepsilon^2}$ and time complexity $O\rbra{n/\varepsilon^2}$.
    
    \begin{proof} [Proof sketch of \cref{corollary:inner-product}]
        This is shown by the ($4$-dimensional) qudit version of \cref{thm:pi-intro} and noting the sample complexity lower bound $\Omega\rbra{1/\varepsilon^2}$ for inner product estimation \cite{ALL22}. 
    \end{proof}

    \paragraph{Time-optimality of Shift test.}

    To estimate high-order functionals of a quantum state such as $\tr\rbra{\rho^k}$ for integer $k \geq 3$, the Shift test, a generalized version of the SWAP test proposed in \cite{EAO+02}, can estimate $\tr\rbra{\rho^k}$ to within additive error $\varepsilon$ using $O\rbra{k/\varepsilon^2}$ samples of $\rho$ and $O\rbra{nk/\varepsilon^2}$ gates if $\rho$ is $n$-qubit. 
    By \cref{thm:pi-intro}, we can establish tight quantum time complexity lower bounds for estimating $\tr\rbra{\rho^k}$, implying that the Shift test is time-optimal for this task. 

    \begin{corollary} [Quantum time lower bounds for high-order power trace estimation, \cref{thm:pow-tr-est}] \label{corollary:pow-tr-intro}
        Given sample access to an $n$-qubit quantum state $\rho$, estimating $\tr\rbra{\rho^k}$ to within additive error $\varepsilon$ requires quantum time complexity $\Omega\rbra{nk/\varepsilon^2}$. 
    \end{corollary}
    \begin{proof} [Proof sketch]
        This can be shown by \cref{thm:pi-intro} and noting the sample complexity lower bound $\Omega\rbra{k/\varepsilon^2}$ for estimating $\tr\rbra{\rho^k}$ \cite{CWYZ25}. 
    \end{proof}

    \paragraph{Time-optimality of product test.}

    For an $n$-partite pure state $\ket{\psi}$, the problem of productness testing is to determine whether $\ket{\psi}$ is an $n$-partite product state or $\varepsilon$-far (in trace distance) from any $n$-partite product states. 
    This problem was first considered in \cite{MKB05}.
    In \cite{HM13}, they presented an efficient tester for this task using $O\rbra{1/\varepsilon^2}$ samples of $\ket{\psi}$ and $O\rbra{nm/\varepsilon^2}$ gates if each of the $n$ parts is $m$-qubit.
    By \cref{thm:full-intro}, we can show that the productness tester in \cite{HM13} is time-optimal. 

    \begin{corollary} [Quantum time lower bounds for productness testing, \cref{thm:prod-test}] \label{corollary:prod-test}
        Given sample access to an $n$-partite quantum state $\ket{\psi}$ with each part consisting of $m$ qubits, determining whether $\ket{\psi}$ is a product state or $\varepsilon$-far from any product states requires quantum time complexity $\Omega\rbra{nm/\varepsilon^2}$. 
    \end{corollary}
    \begin{proof}[Proof sketch]
        This time lower bound is based on the sample complexity lower bound $\Omega\rbra{1/\varepsilon^2}$ for productness testing \cite{SW22,CWZ24}. 
        To apply \cref{thm:full-intro}, we split the $nm$ qubits into two sets $A_0$ and $A_1$, each with $nm/2$ qubits. 
        Then, we show that the smaller problem with $n = 2$ and $m = 1$ is $\mathcal{G}$-invariantly embeddable in this larger problem, where $\mathcal{G} = \Sym\rbra{A_0} \times \Sym\rbra{A_1}$. 
    \end{proof}

    \paragraph{Time-optimality of LMR protocol.}

    The LMR protocol \cite{LMR14,KLL+17,GKP+24} allows us to approximately implement the unitary operator $e^{-i \rho t}$ using samples of $\rho$. 
    In particular, when $\rho$ is an $n$-qubit pure state, we can implement $e^{-i\rho t}$ to precision $\delta$ (in diamond norm) using $O\rbra{1/\delta}$ samples of $\rho$ and $O\rbra{n/\delta}$ two-qubit gates for any real number $t$. 
    By \cref{thm:pi-intro}, we can show that the LMR protocol is time-optimal for pure states. 

    \begin{corollary} [Time-optimality of LMR protocol, \cref{thm:lmr}] \label{corollary:lmr}
        Given sample access to an $n$-qubit pure quantum state $\rho$, implementing $e^{-i\rho t}$ to precision $\delta$ requires quantum time complexity $\Omega\rbra{n/\delta}$, even if $t = \pi$.
    \end{corollary}

    \paragraph{Time-optimality of samplizer.}

    The samplizer \cite{WZ24} is a generalization of the quantum sample-to-query lifting \cite{WZ25}, which allows us to convert a quantum query algorithm, using $Q$ queries to the block-encoding of a quantum state $\rho$, to a quantum algorithm, using $\widetilde{O}\rbra{Q^2/\delta}$ samples of $\rho$ and $\widetilde{O}\rbra{nQ^2/\delta}$ (additional) gates if $\rho$ is $n$-qubit, to precision $\delta$ in the diamond norm distance. 
    In particular, as noted in \cite{WZ24c}, when $\rho$ is a pure state, the block-encoding of $\rho$ is equivalent to the reflection operator $e^{-i\rho \pi}$ that can be implemented by the LMR protocol \cite{LMR14,KLL+17,GKP+24}, and thus the polylogarithmic factors can be removed from the sample and time complexities given in \cite{WZ24}; this observation is also used in quantum (maximum) phase estimation with advice \cite{MdW23}. 
    By \cref{thm:pi-intro}, we can show that the implementation of the samplizer in \cite{WZ24c} is time-optimal for pure states. 

    \begin{corollary} [Time-optimality of samplizer, \cref{thm:time-samplizer}] \label{corollary:samplizer}
        Given sample access to an $n$-qubit pure quantum state $\rho$, for any quantum algorithm using $Q$ queries to the reflection operator $e^{-i\rho\pi}$, its samplized version to precision $\delta$ requires quantum time complexity $\Omega\rbra{nQ^2/\delta}$.
    \end{corollary}

    \cref{corollary:samplizer} also implies that the implementation of samplizer (for mixed states) in \cite{WZ24} is time-optimal up to polylogarithmic factors. 

    \paragraph{Time-optimality of pure-state trace distance estimation.}

    Recently, a sample-optimal quantum estimator for pure-state trace distance and (square root) fidelity was proposed in \cite{WZ24c}, which, for an estimate to within additive error $\varepsilon$, uses $O\rbra{1/\varepsilon^2}$ samples of $n$-qubit pure states and $O\rbra{n/\varepsilon^2}$ two-qubit gates. 
    By \cref{thm:pi-intro}, we can show that the estimator in \cite{WZ24c} is also time-optimal. 

    \begin{corollary} [Quantum time lower bounds for pure-state trace distance estimation, \cref{thm:time-puretd}] \label{corollary:td-estimation}
        Given sample access to two $n$-qubit pure quantum states, estimating their trace distance and (square root) fidelity to within additive error $\varepsilon$ requires quantum time complexity $\Omega\rbra{n/\varepsilon^2}$. 
    \end{corollary}

    \begin{proof} [Proof sketch]
        For (square root) fidelity, the lower bound has been shown in \cref{corollary:inner-product}. 
        For trace distance, the lower bound can be shown by \cref{thm:pi-intro} and noting the sample complexity lower bound $\Omega\rbra{1/\varepsilon^2}$ for trace distance estimation \cite{Wan24}. 
    \end{proof}

    \cref{corollary:td-estimation} also serves as a starting point for proving \cref{corollary:lmr,corollary:samplizer}. 

    \begin{proof} [Proof sketch of \cref{corollary:lmr,corollary:samplizer}]
        This is done by reducing the problem of pure-state trace distance estimation. 
        Let $f_{\textup{LMR}}\rbra{\delta}$ and $f_{\textup{samplize}}\rbra{Q, \delta}$ be the quantum time complexities, respectively, for implementing $e^{-i\rho\pi}$ and for approximating a quantum algorithm with $Q$ queries to the block-encoding of $\rho$, both to precision $\delta$ in diamond norm. 
        Given the quantum algorithm in \cite{WZ24c} that estimates the trace distance between two pure states to within additive error $\varepsilon$ using $O\rbra{1/\varepsilon}$ queries to their reflection operators, this can be done with quantum time complexity $f_{\textup{samplize}}\rbra{O\rbra{1/\varepsilon}, 2/3}$ if given sample access to the pure states. 
        Combining \cref{corollary:td-estimation}, we have established the relation $f_{\textup{samplize}}\rbra{O\rbra{1/\varepsilon}, 2/3} \geq \Omega\rbra{n/\varepsilon^2}$, which gives $f_{\textup{samplize}}\rbra{Q, \delta} \geq \Omega\rbra{nQ^2/\delta}$ with further analyses.  
        To establish the lower bound $f_{\textup{LMR}}\rbra{\delta} \geq \Omega\rbra{n/\delta}$, we only have to note the relation $Q \cdot f_{\textup{LMR}}\rbra{\delta/Q} \geq f_{\textup{samplize}}\rbra{Q, \delta}$ implied in the implementation of the samplizer in \cite{WZ24c}. 
    \end{proof}

    \subsection{Techniques}

    Our main result is based on the light-cone argument. 
    As a warm-up, we first explain how to prove a lower bound for purity testing. 

    \paragraph{An Intuitive Example: Purity Testing.}
    Consider the purity testing problem: determine whether a mixed quantum state $\rho$ is pure or has purity $\tr\rbra{\rho^2} \leq 1 - \varepsilon$. 
    The $1$-qubit case of purity testing can be embedded into the $n$-qubit case by putting this qubit with the other $\rbra{n-1}$ qubits being some pure states. 
    We now want to show that $\Omega\rbra{n/\varepsilon}$ quantum gates are necessary for the $n$-qubit case of purity testing.
    If it is not the case, i.e., there is a tester for purity testing using $o\rbra{n/\varepsilon}$ gates, then we can find a qubit (out of $n$) that is touched only $o\rbra{1/\varepsilon}$ times, which, by the permutation invariance of purity testing, implies a tester for the $1$-qubit case with sample complexity $o\rbra{1/\varepsilon}$.
    On the other hand, this violates the $\Omega\rbra{1/\varepsilon}$ sample lower bound for the $1$-qubit case of purity testing (see \cref{lemma:samp-purity}). 

    The above arguments give an intuitive idea of how to prove \cref{thm:pi-intro,corollary:purity-intro}. 
    \cref{thm:full-intro} further extends this idea to less symmetric cases by carefully finding as few such ``useful'' qubits as possible. 

    \paragraph{Proof Sketch of \cref{thm:full-intro}.}
    From a high-level view, our proof is based on an algorithmic reduction of $\mathcal{Q}_m$ to $\mathcal{P}_n$. 
    Specifically, we are going to construct a tester $\overline{\mathcal{T}}$ for $\mathcal{Q}_m$ with sample complexity $\mathsf{T}\rbra{\mathcal{P}_n}/R$.
    It is noted that the construction of $\overline{\mathcal{T}}$ depends on the actual tester $\mathcal{T}$ for $\mathcal{P}_n$. 
    Here, we assume that the tester $\mathcal{T}$ for $\mathcal{P}_n$ has time complexity $T \geq \mathsf{T}\rbra{\mathcal{P}_n}$ and sample complexity $S$.
    
    To this end, we first figure out a subset of samples that are ``useful''. 
    Let $\sigma$ be a state through which $\mathcal{Q}_m$ is $\mathcal{G}$-invariantly embeddable in $\mathcal{P}_n$. That is, $\rho \in \mathcal{Q}_m^{X}$ if and only if $U_\pi \rbra{\rho \otimes \sigma} U_{\pi}^\dag \in \mathcal{P}_n^X$ for every $\pi \in \mathcal{G}$ and $X \in \cbra{\textup{yes}, \textup{no}}$.
    Then, on input $\rbra{\rho \otimes \sigma}^{\otimes S}$ ($S$ samples of $\rho$ and ignorable $\sigma$), the output of the tester $\mathcal{T}$ can be used to determine whether $\rho \in \mathcal{Q}_m^{\textup{yes}}$ or $\rho \in \mathcal{Q}_m^{\textup{no}}$. 
    In the following, we will show how $\mathcal{T}$ can be modified to a tester for $\mathcal{Q}_m$ with sample complexity $T/R$.

    \begin{enumerate}
        \item For each $1 \leq j \leq k$, we can divide $A_j$ into at least $R$ disjoint subsets $A_j^{\rbra{1}}, A_j^{\rbra{2}}, \dots, A_j^{\rbra{R}}$ such that $A_j^{\rbra{1}} \supseteq A_j \cap \sbra{m}$ and $\abs{A_j^{\rbra{r}}} \geq \abs{A_j^{\rbra{1}}}$ for every $1 \leq r \leq R$ (here we assume that $A_j^{\rbra{1}} \neq \emptyset$ without loss of generality). 
        \item It can be further shown that there exists an $r^*$ such that $A_1^{\rbra{r^*}}, A_2^{\rbra{r^*}}, \dots, A_k^{\rbra{r^*}}$ involve no more than $T/R$ samples of $\rho$; in other words, no more than $T/R$ samples of $\rho$ have at least one of their $n$ qubits that is numbered in $A_1^{\rbra{r^*}} \sqcup A_2^{\rbra{r^*}} \sqcup \dots \sqcup A_k^{\rbra{r^*}} \subseteq \sbra{n}$ and connected to the output qubit.\footnote{Two qubits $q_1$ and $q_2$ are connected in a quantum circuit, if (i) there is a two-qubit gate acting on them, or (ii) there is a two-qubit gate acting on $q_1$ and another qubit $q_3$ such that $q_3$ and $q_2$ are connected.} 
        To see this, for each $1 \leq r \leq R$, let $C_r \subseteq \sbra{S}$ be the set of samples of $\rho$ (out of $S$) with at least one of their qubits that is numbered in $A_1^{\rbra{r}} \sqcup A_2^{\rbra{r}} \sqcup \dots \sqcup A_k^{\rbra{r}}$ and connected to the output qubit. 
        Because of the connectivity of the quantum circuit of the quantum algorithm with time complexity $T$, we have $\sum_{r \in \sbra{R}} \abs{C_r} \leq T$. 
        Therefore, there exists an $r^*$ such that $\abs{C_{r^*}} \leq T/R$. 
        (This corresponds to Part 3 of the proof of \cref{thm:main}.)
    \end{enumerate}

    Now that we have obtained the ``useful'' set $C_{r^*}$ of samples of $\rho$, we can remove the useless samples by the following two steps.

    \begin{enumerate}
    \setcounter{enumi}{2}
        \item Let $\pi^* \in \mathcal{G}$ be the permutation that swaps the elements in $A_{j}^{\rbra{1}}$ and the elements in $A_{j}^{\rbra{r^*}}$ for all $1 \leq j \leq k$. 
        Then, as $\mathcal{Q}_m$ is $\mathcal{G}$-invariantly embeddable in $\mathcal{P}_n$ through $\sigma$, we have that for $X \in \cbra{\textup{yes}, \textup{no}}$, $U_{\pi^*} \rbra{\rho \otimes \sigma} U_{\pi^*}^\dag \in \mathcal{P}_{n}^X$ if and only if $\rho \otimes \sigma \in \mathcal{P}_n^X$. (This corresponds to Part 4 of the proof of \cref{thm:main}.) 
        \item Consider the following input state for $\mathcal{T}$:
        \begin{equation}
            \tilde\rho = \underbrace{\bigotimes_{s \in C_{r^*}} U_{\pi^*} \rbra{\rho \otimes \sigma} U_{\pi^*}^\dag}_{\textup{useful}} \otimes \underbrace{\bigotimes_{s \in \sbra{S} \setminus \sbra{C_{r^*}}} U_{\pi^*} \rbra{\ketbra{\bar0}{\bar0} \otimes \sigma} U_{\pi^*}^\dag}_{\textup{useless}}.
        \end{equation}
        It can be shown that the output of $\mathcal{T}$ on input $\tilde\rho$ obeys the same probability distribution as the output of $\mathcal{T}$ on input $U_{\pi^*}\rbra{\rho\otimes\sigma}U_{\pi^*}^\dag$. (This corresponds to Part 5 of the proof of \cref{thm:main}.)
    \end{enumerate}

    The proof is completed by noting that $T/R$ samples of $\rho$ suffice to prepare $\tilde\rho$. 
    This is simple, as $\tilde\rho$ can be obtained by performing $U_{\pi^*}^{\otimes S}$ on (re-ordered) $\rho^{\otimes \abs{C_{r^*}}} \otimes \sigma^{\otimes S} \otimes \ketbra{\bar0}{\bar0}^{\otimes \rbra{S-\abs{C_{r^*}}}}$, which only uses $\abs{C_{r^*}} \leq T/R$ samples of $\rho$. 
    Note that $\sigma$ and $\ket{\bar0}$ are independent of $\rho$. 

    \subsection{Related Work}

    The quantum time complexity for quantum query algorithms has recently been investigated in the literature \cite{CJOP20,BJY24,ABB+23}, as well as conditional lower bounds on quantum time complexity \cite{ACL+20,BPS21} related to the quantum strong exponential-time hypotheses.

    Permutation invariance is also useful for quantum state tomography \cite{TWG+10}, quantum error corrections \cite{PR04,Ouy14}, and the quantum complexity of Boolean functions \cite{AA14,Cha19,BDCG+24,GHYY25}. 
\iffalse
    Quantum time complexity for quantum query algorithms. 
    \url{https://arxiv.org/abs/2311.16401}
    \url{https://quantum-journal.org/papers/q-2024-08-23-1444/}
    \url{https://arxiv.org/abs/2005.01323}

    Conditional Quantum time lower bound. 
    \url{https://drops.dagstuhl.de/entities/document/10.4230/LIPIcs.CCC.2020.16}
    \url{https://drops.dagstuhl.de/entities/document/10.4230/LIPIcs.STACS.2021.19}

    PI code. \url{https://arxiv.org/pdf/quant-ph/0304153}

    Understanding PI. \url{https://arxiv.org/pdf/quant-ph/0301020}

    PI tomography. \url{https://arxiv.org/abs/1005.3313}

    PI Property quantum book \url{https://doi.org/10.1017/9781316848142.008}

    permunation-invariant functions. \url{https://arxiv.org/abs/2401.00454}
\fi
    
    \subsection{Discussion}

    In this paper, we proposed a method for proving lower bounds on the quantum time complexity for quantum property testing. 
    This method is especially useful for permutation-invariant properties. 
    Several tight lower bounds are obtained through this method, showing the time-optimality of a series of quantum algorithmic tools such as the SWAP test \cite{BCWdW01}, Shift test \cite{EAO+02}, multipartite productness test \cite{HM13}, LMR protocol \cite{LMR14}, and samplizer \cite{WZ24}. 
    In the following, we provide some questions for future research. 
    \begin{enumerate}
        \item Sample-optimal approaches to testing the mixedness and rank of an unknown quantum state are presented in \cite{CHW07,OW21}. Since these properties are also permutation-invariant, a meaningful question is: can we find a time-optimal approach to testing them? 
        The current approaches are based on weak Schur sampling \cite{CHW07}, and thus its time complexity has a polynomial overhead compared to its sample complexity. 
        It is also interesting to consider other property testing problems. 
        \item Time-efficiency is important in all cases of quantum computing. 
        Can we prove the time-optimality of any other existing quantum tools or can we develop new quantum tools that are time-optimal?
        \item We hope that our discovery can inspire further work on lower bounding quantum time complexity. 
        A central question is: can we extend the sample-to-time reduction in \cref{thm:full-intro} to a broader range of properties? 
        In addition, can we develop other methods for proving lower bounds on quantum time complexity?
    \end{enumerate}

    \section{Preliminaries}

    In this section, we define necessary notions for quantum state testing. 
    
    \paragraph{Basic notations.}
    We use $\sbra{n} = \cbra{1, 2, \dots, n}$. 
    In particular, $\sbra{0} = \varnothing$ is the empty set.
    Let $\Sym\rbra{A}$ denote the symmetric group over the set $A$ and we use the shorthand $\Sym\rbra{n} \coloneqq \Sym\rbra{\sbra{n}}$. 
    For two groups $\mathcal{G}$ and $\mathcal{G'}$, we denote $\mathcal{G}' \leq \mathcal{G}$ to mean that $\mathcal{G}'$ is a subgroup of $\mathcal{G}$. 

    \subsection{Properties of quantum states}

    Let $\mathcal{H}_2$ be the $2$-dimensional Hilbert space of qubits, where a qubit is described by a linear combination of $\ket{0}$ and $\ket{1}$. 
    The concept of qubits can be extended to qudits, where a qudit is in the $d$-dimensional Hilbert space $\mathcal{H}_d$ for some $d$, i.e., a linear combination of $\ket{0}$, $\ket{1}$, \dots, $\ket{d-1}$.
    Throughout this paper, all the concepts with respect to qubits can be naturally extended to qudits.
    For simplicity, we mainly consider the case of qubits.

    Let $\mathcal{D}\rbra{\mathcal{H}}$ be the set of density operators (or, equivalently, mixed quantum states) on $\mathcal{H}$, i.e., $\mathcal{D}\rbra{\mathcal{H}}$ consists of all the positive semidefinite operators $\rho$ on $\mathcal{H}$ satisfying $\tr\rbra{\rho} = 1$. 
    The trace distance and fidelity between the two mixed states $\rho, \sigma \in \mathcal{D}\rbra{\mathcal{H}}$ are, respectively, defined by 
    \begin{equation}
        \mathrm{T}\rbra{\rho, \sigma} = \frac{1}{2}\tr\rbra{\abs{\rho - \sigma}}, \qquad \mathrm{F}\rbra{\rho, \sigma} = \tr\rbra*{\sqrt{\sqrt{\sigma}\rho\sqrt{\sigma}}}.
    \end{equation}

    A property of $n$-qubit mixed quantum states (``$n$-qubit property'' for short), denoted as $\mathcal{P}_n = \rbra{\mathcal{P}_n^{\textup{yes}}, \mathcal{P}_n^{\textup{no}}} \subseteq \mathcal{D}\rbra{\mathcal{H}_2^{\otimes n}} \times \mathcal{D}\rbra{\mathcal{H}_2^{\otimes n}}$, is a pair of disjoint sets of $n$-qubit mixed quantum states.
    The definition of permutation-invariant properties is given as follows. 

    \begin{definition} [Permutation-invariant properties] \label{def:perm-inv}
        An $n$-qubit property $\mathcal{P}_n$ is said to be permutation-invariant, if for every permutation $\pi \in \Sym\rbra{n}$ and $X \in \cbra{\textup{yes}, \textup{no}}$, $U_{\pi} \rho U_{\pi}^\dag \in \mathcal{P}_n^X$ if and only if $\rho \in \mathcal{P}_n^X$, where 
        \begin{equation} \label{eq:def-U-pi}
        U_\pi \colon \ket{\psi_1}\ket{\psi_2}\dots\ket{\psi_n} \mapsto \ket{\psi_{\pi\rbra{1}}}\ket{\psi_{\pi\rbra{2}}}\dots\ket{\psi_{\pi\rbra{n}}}
        \end{equation}
        for every $\ket{\psi_1}, \ket{\psi_2}, \dots, \ket{\psi_n} \in \mathcal{H}_2$.
    \end{definition}

    We call a property $\mathcal{Q}$ embeddable in another property $\mathcal{P}$, if $\mathcal{Q}$ can be considered as a special case of $\mathcal{P}$ so that $\mathcal{P}$ is harder than $\mathcal{Q}$ by reduction.
    Formally, we have the following definition. 

    \begin{definition} [Embeddability] \label{def:embed}
        We say that an $m$-qubit property $\mathcal{Q}_m$ is embeddable in another $n$-qubit property $\mathcal{P}_n$ with $m < n$, denoted as $\mathcal{Q}_m \hookrightarrow \mathcal{P}_n$, if there is an $\rbra{n-m}$-qubit mixed state $\sigma$ such that for every $\rho \in \mathcal{Q}_m^X$ and $X \in \cbra{\textup{yes}, \textup{no}}$, we have $\rho \otimes \sigma \in \mathcal{P}_n^X$.
        When the state $\sigma$ should be made clear from the context, we write $\mathcal{Q}_m \xhookrightarrow[]{\sigma} \mathcal{P}_n$. 
    \end{definition}

    To extract the necessary conditions that are enough for our results, we characterize a weaker type of embeddability in terms of permutation groups. 

    \begin{definition} [Group-invariant embeddability] \label{def:inv-embed}
        Let $\mathcal{P}_n$ and $\mathcal{Q}_m$ be $n$- and $m$-qubit properties, respectively, with $m < n$. 
        Let $\mathcal{G} \leq \Sym\rbra{n}$ be a permutation group. 
        For an $\rbra{n-m}$-qubit mixed quantum state $\sigma$, $\mathcal{Q}_m$ is said to be $\mathcal{G}$-invariantly embeddable in $\mathcal{P}_n$ through $\sigma$, 
        denoted as $\mathcal{Q}_m \xhookrightarrow[\mathcal{G}]{\sigma} \mathcal{P}_n$, if $U_\pi \rbra{\rho \otimes \sigma} U_\pi^\dag \in \mathcal{P}_n^X$ for every $\rho \in \mathcal{Q}_m^X$, $\pi \in \mathcal{G}$, and $X \in \cbra{\textup{yes}, \textup{no}}$, where $U_\pi$ is defined by \cref{eq:def-U-pi}. 
    \end{definition}

    Embeddability with permutation invariance can be seen as a special case of group-invariant embeddability, shown as follows. 

    \begin{fact} [Embeddability with permutation invariance]
        If $\mathcal{Q}_m \xhookrightarrow[]{\sigma} \mathcal{P}_n$ and $\mathcal{P}_n$ is permutation-invariant, then $\mathcal{Q}_m \xhookrightarrow[\Sym\rbra{n}]{\sigma} \mathcal{P}_n$.
    \end{fact}

    \subsection{Testers} \label{sec:tester}

    Suppose that $\mathcal{T} = \rbra{\mathcal{T}_1, \mathcal{T}_2, \dots, \mathcal{T}_n, \dots}$ is a family of (non-uniform) testers, where $\mathcal{T}_n$ is a tester for $n$-qubit states described by a quantum unitary circuit acting on $\rbra{\mathcal{H}_2^{\otimes n}}^{\otimes S\rbra{n}} \otimes \mathcal{H}_2^{\otimes \ell\rbra{n}}$ for some functions $S, \ell \colon \mathbb{N} \to \mathbb{N}$.
    Specifically, $\mathcal{T}_n$ has the form
    \begin{equation}
        \mathcal{T}_n = U_{T\rbra{n}-1} \cdot \cdots \cdot U_2 \cdot U_1, 
    \end{equation}
    where $U_t$ for $1 \leq t \leq T\rbra{n}-1$ is a two-qubit unitary gate for some function $T \colon \mathbb{N} \to \mathbb{N}$.\footnote{The reason that there are $\rbra{T\rbra{n}-1}$ two-qubit gates but not $T\rbra{n}$ is that the last operation of the tester is fixed to perform a quantum measurement on the first qubit in the computational basis, which theoretically should be considered to take one unit of time. With this definition, our results can be written in a graceful form.}
    We call $\mathcal{T}$ a family of testers with sample complexity $S$, time complexity $T$, and auxiliary space complexity $\ell$. 
    For every $\rho \in \mathcal{D}\rbra{\mathcal{H}_2^{\otimes n}}$, the probability that $\mathcal{T}_n$ accepts $\rho$ is defined by 
    \begin{equation} \label{eq:def-T-prob}
        \Pr\sbra{\mathcal{T}_n \text{ accepts } \rho} = \tr\rbra*{ \Pi \mathcal{T}_n \rbra*{ \rho^{\otimes S\rbra{n}} \otimes \ketbra{0}{0}^{\otimes \ell\rbra{n}} } \mathcal{T}_n^\dag }, 
    \end{equation}
    where $\Pi = \ketbra{0}{0} \otimes I_2^{\otimes \rbra{n S\rbra{n} + \ell\rbra{n} - 1}}$ is the projector onto the subspace of $\mathcal{H}_2^{\otimes \rbra{nS\rbra{n} + \ell\rbra{n}}}$ with the first qubit being $\ket{0}$ and $I_2$ is the identity operator on $\mathcal{H}_2$. 

    Let $\mathcal{P} = \rbra{\mathcal{P}_1, \mathcal{P}_2, \dots, \mathcal{P}_n, \dots}$ be a family of properties of mixed quantum states. 
    Tester $\mathcal{T}_n$ is said to be a tester for $\mathcal{P}_n$, if for every $\rho^\textup{yes} \in \mathcal{P}_n^{\textup{yes}}$, $\Pr\sbra{\mathcal{T}_n \text{ accepts } \rho^{\textup{yes}}} \geq 2/3$, and for every $\rho^\textup{no} \in \mathcal{P}_n^{\textup{no}}$, $\Pr\sbra{\mathcal{T}_n \text{ accepts } \rho^{\textup{no}}} \leq 1/3$. 
    The sample complexity of $\mathcal{P}_n$, denoted as $\mathsf{S}\rbra{\mathcal{P}_n}$, is the minimum sample complexity of $\mathcal{T}_n$ over all testers $\mathcal{T}_n$ for $\mathcal{P}_n$. 
    The time complexity of $\mathcal{P}_n$, denoted as $\mathsf{T}\rbra{\mathcal{P}_n}$, is the minimum time complexity of $\mathcal{T}_n$ over all testers $\mathcal{T}_n$ for $\mathcal{P}_n$. 

    \section{Sample-to-Time Reduction for Permutation-Invariant Properties}

    In this section, we state our main theorem as follows. 

    \begin{theorem}[Sample-to-time reduction for partially permutation-invariant embeddability] \label{thm:main}
        Let $\mathcal{P}_n$ and $\mathcal{Q}_m$ be $n$- and $m$-qubit properties, respectively, with $1 \leq m < n$.
        Let $\mathcal{G} = \Sym\rbra{A_{1}} \times \Sym\rbra{A_{2}} \times \dots \times \Sym\rbra{A_{k}} \leq \Sym\rbra{n}$ be a permutation group, where $A_1, A_2, \dots, A_k$ form a partition of $\sbra{n}$. 
        If $\mathcal{Q}_m \xhookrightarrow[\mathcal{G}]{\sigma} \mathcal{P}_n$, then
        \begin{equation}
            \mathsf{T}\rbra{\mathcal{P}_n} \geq R \cdot \mathsf{S}\rbra{\mathcal{Q}_m},
        \end{equation}
        where
        \begin{equation}
            R = \min_{j \in \sbra{k} \colon A_j \cap \sbra{m} \neq \emptyset} \floor*{\frac{\abs{A_j}}{\abs*{A_j \cap \sbra{m}}}}.
        \end{equation}
    \end{theorem}

    As a special case of \cref{thm:main}, we have a simple form for permutation-invariant properties. 

    \begin{theorem} [Sample-to-time reduction for permutation-invariant properties] \label{thm:main-sp}
        Let $\mathcal{P}_n$ be an $n$-qubit permutation-invariant property and $\mathcal{Q}_1$ be a $1$-qubit property. 
        If $\mathcal{Q}_1 \hookrightarrow \mathcal{P}_n$, then 
        \begin{equation}
            \mathsf{T}\rbra{\mathcal{P}_n} \geq n \cdot \mathsf{S}\rbra{\mathcal{Q}_1}. 
        \end{equation}
    \end{theorem}
    \begin{proof}
        This is immediately obtained by \cref{thm:main} with $m = k = 1$ and $A_1 = \sbra{n}$. 
        Then, by simple calculations, we have $R = n$. 
    \end{proof}

    \begin{remark} [Necessity of permutation invariance]
        Some readers may wonder if the permutation-invariant condition is necessary for the ratio $R$ (of time complexity to sample complexity) to be at least $n$, the number of qubits, as in \cref{thm:pi-intro}. As shown in \cref{fn:tight-example}, the ratio can be as low as a constant in the presence of useless qubits. 
        Actually, even if there is no useless qubit, the ratio can be as low as $\polylog\rbra{n}$ (see \cref{sec:low-ratio}). 
        Here, useless qubits mean those that are not involved in the definition of the property of quantum states.
    \end{remark}

    Note that \cref{thm:main} (and thus \cref{thm:main-sp}) can also hold for the qudit case by the same arguments.
    For completeness, we present the qudit version of \cref{thm:main,thm:main-sp} as follows.

    \begin{theorem}[Qudit version of \cref{thm:main}] \label{thm:main-qudit}
        Let $\mathcal{P}_n$ and $\mathcal{Q}_m$ be $n$- and $m$-qudit properties, respectively, with $1 \leq m < n$.
        Let $\mathcal{G} = \Sym\rbra{A_{1}} \times \Sym\rbra{A_{2}} \times \dots \times \Sym\rbra{A_{k}} \leq \Sym\rbra{n}$ be a permutation group, where $A_1, A_2, \dots, A_k$ form a partition of $\sbra{n}$. 
        If $\mathcal{Q}_m \xhookrightarrow[\mathcal{G}]{\sigma} \mathcal{P}_n$, then
        \begin{equation}
            \mathsf{T}\rbra{\mathcal{P}_n} \geq R \cdot \mathsf{S}\rbra{\mathcal{Q}_m},
        \end{equation}
        where
        \begin{equation}
            R = \min_{j \in \sbra{k} \colon A_j \cap \sbra{m} \neq \emptyset} \floor*{\frac{\abs{A_j}}{\abs*{A_j \cap \sbra{m}}}}.
        \end{equation}
    \end{theorem}

    \begin{theorem} [Qudit version of \cref{thm:main-sp}] \label{thm:main-sp-qudit}
        Let $\mathcal{P}_n$ be an $n$-qudit permutation-invariant property and $\mathcal{Q}_1$ be a $1$-qudit property. 
        If $\mathcal{Q}_1 \hookrightarrow \mathcal{P}_n$, then 
        \begin{equation}
            \mathsf{T}\rbra{\mathcal{P}_n} \geq n \cdot \mathsf{S}\rbra{\mathcal{Q}_1}. 
        \end{equation}
    \end{theorem}

    The remainder of this section is the proof of \cref{thm:main}.

    \subsection{Proof of Theorem \ref{thm:main}}

    The proof of \cref{thm:main} consists of 8 steps. 

    \paragraph{Part 1: Basic set-up.}
    Let $\mathcal{P}_n = \rbra{\mathcal{P}_n^{\textup{yes}}, \mathcal{P}_n^{\textup{no}}}$. 
    Suppose that $\mathcal{T}$ is a tester for $\mathcal{P}_n$ with sample complexity $S$, time complexity $T = \mathsf{T}\rbra{\mathcal{P}_n}$, and auxiliary space complexity $\ell$. 
    That is, for every $\rho \in \mathcal{D}\rbra{\mathcal{H}_2^{\otimes n}}$, the probability that $\mathcal{T}$ accepts $\rho$ is defined by 
    \begin{equation} \label{def:Tester}
        \Pr\sbra{\mathcal{T} \text{ accepts } \rho} = \tr\rbra*{ \Pi \mathcal{T} \rbra*{ \rho^{\otimes S} \otimes \ketbra{0}{0}^{\otimes \ell} } \mathcal{T}^\dag }, 
    \end{equation}
    where $\Pi = \ketbra{0}{0} \otimes I_2^{\otimes \rbra{n S + \ell - 1}}$ is the projector onto the subspace of $\mathcal{H}_2^{\otimes \rbra{nS + \ell}}$ with the first qubit being $\ket{0}$ and $I_2$ is the identity operator on $\mathcal{H}_2$. 
    We number the qubits from $1$ to $nS + \ell$ as follows:
    \begin{itemize}
        \item The $x$-th qubit of the $s$-th sample of the input state $\rho$ is called qubit $\rbra{s - 1} n + x$ for $s \in \sbra{S}$ and $x \in \sbra{n}$. 
        \item The $j$-th ancilla qubit is called qubit $nS + j$ for $j \in \sbra{\ell}$.
    \end{itemize}
    Specifically, $\mathcal{T}$ can be described as a quantum circuit acting on $nS+\ell$ qubits and consisting of $T - 1$ two-qubit unitary gates:
    \begin{equation}
        \mathcal{T} = U_{T-1} \cdot \cdots \cdot U_2 \cdot U_1,
    \end{equation}
    where $U_t$ acts on the $x_t$-th and $y_t$-th qubits with $1 \leq x_t < y_t \leq nS+\ell$ for $t \in \sbra{T-1}$. 

    \paragraph{Part 2: Embedding.}
    Since $\mathcal{Q}_m \xhookrightarrow[\mathcal{G}]{\sigma} \mathcal{P}_n$ with $\sigma$ an $\rbra{n-m}$-qubit state, we have $\rho \otimes \sigma \in \mathcal{P}_n^{X}$ for every $\rho \in \mathcal{Q}_m^X$ and $X \in \cbra{\textup{yes}, \textup{no}}$. 
    By the definition, the tester $\mathcal{T}$ accepts $\rho \otimes \sigma$ with probability 
    \begin{equation} \label{eq:T-prob}
        \Pr\sbra{\mathcal{T} \textup{ accepts } \rho \otimes \sigma} = \tr\rbra*{\Pi \mathcal{T} \rbra*{\rbra{\rho \otimes \sigma}^{\otimes S} \otimes \ketbra{0}{0}^{\otimes \ell}} \mathcal{T}^\dag}. 
    \end{equation}
    Moreover,
    \begin{itemize}
        \item $\Pr\sbra{\mathcal{T} \textup{ accepts } \rho \otimes \sigma} \geq 2/3$ if $\rho \in \mathcal{Q}_m^{\textup{yes}}$, 
        \item $\Pr\sbra{\mathcal{T} \textup{ accepts } \rho \otimes \sigma} \leq 1/3$ if $\rho \in \mathcal{Q}_m^{\textup{no}}$.
    \end{itemize}

    For clarification, we divide qubits into groups:
    \begin{itemize}
        \item $\mathsf{X}_s = \set{\rbra{s-1}n+x}{x \in \sbra{n}}$ for $s \in \sbra{S}$. In particular, we write $\mathsf{X}_{s, x} = \cbra{\rbra{s-1}n+x}$ to denote qubit $\rbra{s-1}n+x$ for $s \in \sbra{S}$ and $x \in \sbra{n}$. 
        \item $\mathsf{W} = \set{nS+j}{j \in \sbra{\ell}}$.
    \end{itemize}
    Then, \cref{eq:T-prob} can be written as
    \begin{equation}
        \Pr\sbra{\mathcal{T} \textup{ accepts } \rho \otimes \sigma} = \tr\rbra*{\Pi_\mathsf{XW} \mathcal{T}_\mathsf{XW} \rbra*{ \bigotimes_{s \in \sbra{S}} \rbra*{\rho \otimes \sigma}_{\mathsf{X}_s} \otimes \ketbra{0}{0}^{\otimes \ell}_{\mathsf{W}}} \mathcal{T}_\mathsf{XW}^\dag}, 
    \end{equation}
    where
    \begin{equation} \label{eq:def-Pi-XYW}
        \Pi_{\mathsf{XW}} = \ketbra{0}{0}_{\mathsf{X}_{1,1}} \otimes I_{\mathsf{X}_{1, 2} \dots \mathsf{X}_{1, n}} \otimes I_{\mathsf{X}_2 \dots \mathsf{X}_S\mathsf{W}}. 
    \end{equation}
    Here, the subscripts denote the registers on which an operator acts, and sometimes the subscripts can be ignored if they are clear from the context. 

    \paragraph{Part 3: Connectivity.}
    The connectivity between qubits can be described by an undirected graph $G = \rbra{V, E}$ (possibly with duplicated edges), where
    \begin{equation} \label{eq:def-G}
        V = \sbra{nS+\ell}, \qquad E = \set{\cbra{x_t, y_t}}{t \in \sbra{T-1}}.
    \end{equation}

    We list the elements in each $A_j$ for $j \in \sbra{k}$ in ascending order:
    \begin{equation}
        A_j = \set{a_{j, p}}{p \in \sbra{\abs{A_j}}} \textup{ with } a_{j, 1} < a_{j, 2} < \dots < a_{j, \abs{A_j}}. 
    \end{equation}
    Let
    \begin{equation}
        R = \min_{j \in \sbra{k}} \floor*{\frac{\abs{A_{j}}}{m_{j}}}, \textup{ where } m_{j} = \abs*{A_{j} \cap \sbra{m}}.
    \end{equation}
    For each $j \in \sbra{k}$, we select $R$ disjoint subsets of qubits from $A_{j}$, with the $r$-th subset chosen by 
    \begin{equation}
        A_{j}^{\rbra{r}} = \set{a_{j, p}}{\rbra{r-1}m_j < p \leq r m_j}.
    \end{equation}
    Let $c_r = \abs{C_r}$, where
    \begin{equation}
        C_r = \set{s \in \sbra{S}}{\textup{vertices $\rbra{s-1}n+x$ and $1$ are connected in $G$ for some $x \in A_{j}^{\rbra{r}}$ and $j \in \sbra{k}$}}.
    \end{equation}
    Here, two vertices $u$ and $v$ are connected, if there is a sequence $u_0, u_1, \dots, u_{d}$ such that (i) $u_0 = u$, (ii) $u_d = v$, and (iii) $\cbra{u_{i-1}, u_{i}} \in E$ for all $i \in \sbra{d}$. 
    
    The summation over $c_r$ is upper bounded by the number of vertices in $G$ that are connected to vertex $1$. As there are $T - 1$ edges in $G$, we have
    \begin{equation}
        \sum_{r \in \sbra{R}} c_r \leq T. 
    \end{equation}
    By the Pigeonhole Principle, there exists an $r^* \in \sbra{R}$ such that 
    \begin{equation} \label{eq:c-vs-T/R}
        c_{r^*} \leq \frac{T}{R}. 
    \end{equation}

    \vspace{10pt}
    Now we are going to construct a tester for $\mathcal{Q}_m$ from the implementation of $\mathcal{T}$ with sample complexity $c_{r^*}$.
    
    \paragraph{Part 4: Permutation invariance.} 
    We choose the following permutation over $\sbra{n}$:
    \begin{equation} \label{def:pi}
        \pi = \prod_{j \in \sbra{k}} \pi_j \in \mathcal{G}, 
    \end{equation}
    where
    \begin{equation}
        \pi_j = \prod_{p \in \sbra{m_j}} \rbra{a_{j, p}~a_{j, \rbra{r^*-1}m_j+p}} \in \Sym\rbra{A_j}.
    \end{equation}
    Note that $\pi^2 = \rbra{1}\rbra{2} \dots \rbra{n}$ is the identity permutation, and thus $\pi = \pi^{-1}$. 
    Because of the permutation invariance that $\mathcal{Q}_m \xhookrightarrow[\mathcal{G}]{\sigma} \mathcal{P}_n$, we have 
    \begin{itemize}
        \item $\Pr\sbra{\mathcal{T} \textup{ accepts } U_{\pi}\rbra{\rho \otimes \sigma}U_{\pi}^\dag} \geq 2/3$ if $\rho \in \mathcal{Q}_m^{\textup{yes}}$, 
        \item $\Pr\sbra{\mathcal{T} \textup{ accepts } U_{\pi}\rbra{\rho \otimes \sigma}U_{\pi}^\dag} \leq 1/3$ if $\rho \in \mathcal{Q}_m^{\textup{no}}$,
    \end{itemize}
    where $U_\pi \colon \ket{\psi_1} \ket{\psi_2} \dots \ket{\psi_{n}} \mapsto \ket{\psi_{\pi\rbra{1}}} \ket{\psi_{\pi\rbra{2}}} \dots \ket{\psi_{\pi\rbra{n}}}$ for every $\ket{\psi_1}, \ket{\psi_2}, \dots, \ket{\psi_n} \in \mathcal{H}_2$.
    Note that 
    \begin{equation} \label{eq:prob-T-pi}
        \Pr\sbra*{\mathcal{T} \textup{ accepts } U_{\pi}\rbra{\rho \otimes \sigma}U_{\pi}^\dag} = \tr\rbra*{\Pi \mathcal{T} \rbra*{ \bigotimes_{s \in \sbra{S}} \rbra*{U_{\pi}\rbra{\rho \otimes \sigma}U_{\pi}^\dag}_{\mathsf{X}_s} \otimes \ketbra{0}{0}^{\otimes \ell}_{\mathsf{W}}} \mathcal{T}^\dag}. 
    \end{equation}

    \paragraph{Part 5: Instance encoding with tester $\widetilde{\mathcal{T}}$.}
    Now we consider the state 
    \begin{equation} \label{def:eta}
        \eta_{\mathsf{X}} = \bigotimes_{s \in C_{r^*}} \rbra*{\rho \otimes \sigma}_{\mathsf{X}_s} \otimes \bigotimes_{s \in \sbra{S} \setminus C_{r^*}} \rbra*{\ketbra{0}{0}^{\otimes m} \otimes \sigma}_{\mathsf{X}_s}.
    \end{equation}
    Let 
    \begin{equation} \label{def:tildeT}
        \widetilde{\mathcal{T}} = \mathcal{T} \cdot \rbra*{ \bigotimes_{s \in \sbra{S}} \rbra{U_{\pi}}_{\mathsf{X}_s} \otimes I_{\mathsf{W}} }. 
    \end{equation}
    Then, it can be shown that (see \cref{lemma:prob-tildeT}) 
    \begin{equation}
        \tr\rbra*{\Pi \widetilde{\mathcal{T}} \rbra*{ \eta_{\mathsf{X}} \otimes \ketbra{0}{0}^{\otimes \ell}_{\mathsf{W}}} \widetilde{\mathcal{T}}^\dag} = \Pr\sbra*{\mathcal{T} \textup{ accepts } U_{\pi}\rbra{\rho \otimes \sigma}U_{\pi}^\dag}.
    \end{equation}

    \paragraph{Part 6: Instance encoding with tester $\widehat{\mathcal{T}}$.} 
    We write
    \begin{equation}
        C_{r^*} = \set{s_{j}}{j \in \sbra{c_{r^*}}}.
    \end{equation}
    Let $\tau \in \Sym\rbra{S}$ be a permutation such that $\tau\rbra{C_{r^*}} = \sbra{c_{r^*}}$, i.e., $\tau\rbra{s_j} = j$ for all $j \in \sbra{c_{r^*}}$. 
    Then, let
    \begin{equation} \label{def:hatT}
        \widehat{\mathcal{T}} = \widetilde{\mathcal{T}} \cdot \rbra*{\rbra{V_\tau}_{\mathsf{X}} \otimes I_{\mathsf{W}}},
    \end{equation}
    where 
    \begin{equation}
        V_{\tau} \ket{\phi_1} \ket{\phi_2} \dots \ket{\phi_S} = \ket{\phi_{\tau\rbra{1}}} \ket{\phi_{\tau\rbra{2}}} \dots \ket{\phi_{\tau\rbra{S}}}
    \end{equation}
    for any $\ket{\phi_1}, \ket{\phi_2}, \dots, \ket{\phi_S} \in \mathcal{H}_2^{\otimes n}$; or, equivalently, 
    \begin{equation}
        \rbra{V_{\tau}}_{\mathsf{X}} \cdot \bigotimes_{s \in \sbra{S}} \ket{\phi_s}_{\mathsf{X}_{s}} = \bigotimes_{s \in \sbra{S}} \ket{\phi_{\tau\rbra{s}}}_{\mathsf{X}_{s}} = \bigotimes_{s \in \sbra{S}} \ket{\phi_{s}}_{\mathsf{X}_{\tau^{-1}\rbra{s}}}. 
    \end{equation}
    Then, it can be shown that (see \cref{lemma:prob-tildeT=prob-hatT})
    \begin{equation}
        \tr\rbra*{\Pi \widehat{\mathcal{T}} \rbra*{ \bigotimes_{s \in \sbra{c_{r^*}}} \rbra*{\rho \otimes \sigma}_{\mathsf{X}_s} \otimes \bigotimes_{s = c_{r^*}+1}^{S} \rbra*{\ketbra{0}{0}^{\otimes m} \otimes \sigma}_{\mathsf{X}_s} \otimes \ketbra{0}{0}^{\otimes \ell}_{ \mathsf{W}} } \widehat{\mathcal{T}}^\dag} = \tr\rbra*{\Pi \widetilde{\mathcal{T}} \rbra*{ \eta_{\mathsf{X}} \otimes \ketbra{0}{0}^{\otimes \ell}_{\mathsf{W}}} \widetilde{\mathcal{T}}^\dag}.
    \end{equation}

    \begin{figure} [!htp]
    \centering
    \begin{quantikz}
        \lstick[6]{$\mathsf{X}_1$} \qquad \qquad \qquad \qquad \quad & \lstick{\scriptsize$1$} \wireoverride{n} & & & \gate[17]{\quad V_\tau \quad} \gategroup[23,steps=3,style={dashed,rounded corners,fill=blue!20,inner sep=6pt},background]{\textcolor{blue}{$\widehat{\mathcal{T}}$}} & \gate[6]{~ U_\pi ~} \gategroup[23,steps=2,style={dashed,rounded corners,fill=red!20,inner sep=1pt},background]{\textcolor{red}{$\widetilde{\mathcal{T}}$}} & \gate[23]{\qquad \qquad \mathcal{T} \qquad \qquad} & \meter{} \\
        & \lstick{\scriptsize\vdots} \setwiretype{n} \\
        & \lstick{\scriptsize$m$} \wireoverride{n} & & & & & & \\
        & \lstick{\scriptsize$m+1$} \wireoverride{n} & \gate[6]{\mathcal{O}_\sigma} & & & & & \\
        & \lstick{\scriptsize\vdots} \setwiretype{n} & & & & \\
        & \lstick{\scriptsize$n$} \wireoverride{n} & & & & & & \\
        \lstick[3]{$\mathsf{Z}_1$} \qquad \qquad \qquad \qquad \quad & \lstick{\scriptsize$Sn+\ell+1$} \wireoverride{n} & & \meterD{} \\
        \setwiretype{n} & \lstick{\scriptsize\vdots}  & \\
        & \lstick{\scriptsize$Sn+\ell+n-m$} \wireoverride{n} & & \meterD{} \\
        \lstick[2]{\huge\vdots} \qquad \qquad \qquad \qquad \quad & \lstick{\vdots} \setwiretype{n} &  \\
        & \setwiretype{n} & \\
        \lstick[6]{$\mathsf{X}_S$} \qquad \qquad \qquad \qquad \quad & \lstick{\scriptsize$\rbra{S-1}n+1$} \wireoverride{n} & & & & \gate[6]{~U_\pi~} & & \\
        & \lstick{\scriptsize\vdots} \setwiretype{n} \\
        & \lstick{\scriptsize$\rbra{S-1}n+m$} \wireoverride{n} & & & & & & \\
        & \lstick{\scriptsize$\rbra{S-1}n+m+1$} \wireoverride{n} & \gate[6]{\mathcal{O}_\sigma} & & & & & \\
        & \lstick{\scriptsize\vdots} \setwiretype{n} & & & & & & \\
        & \lstick{\scriptsize$Sn$} \wireoverride{n} & & & & & & \\
        \lstick[3]{$\mathsf{Z}_S$} \qquad \qquad \qquad \qquad \quad & \lstick{\scriptsize$Sn+\ell+\rbra{S-1}\rbra{n-m}+1$} \wireoverride{n} & & \meterD{} \\
        \setwiretype{n} & \lstick{\scriptsize\vdots}  & \\
        & \lstick{\scriptsize$Sn+\ell+S\rbra{n-m}$} \wireoverride{n} & & \meterD{} \\
        \lstick[3]{$\mathsf{W}$} \qquad \qquad \qquad \qquad \quad & \lstick{\scriptsize$Sn+1$} \wireoverride{n} & & & & & & \\
        & \lstick{\scriptsize\vdots} \setwiretype{n} & & & & & \\
        & \lstick{\scriptsize$Sn+\ell$} \wireoverride{n} & & & & & & 
    \end{quantikz}
    \caption{Quantum circuit for tester $\overline{\mathcal{T}}$.}
    \label{fig:view}
    \end{figure}
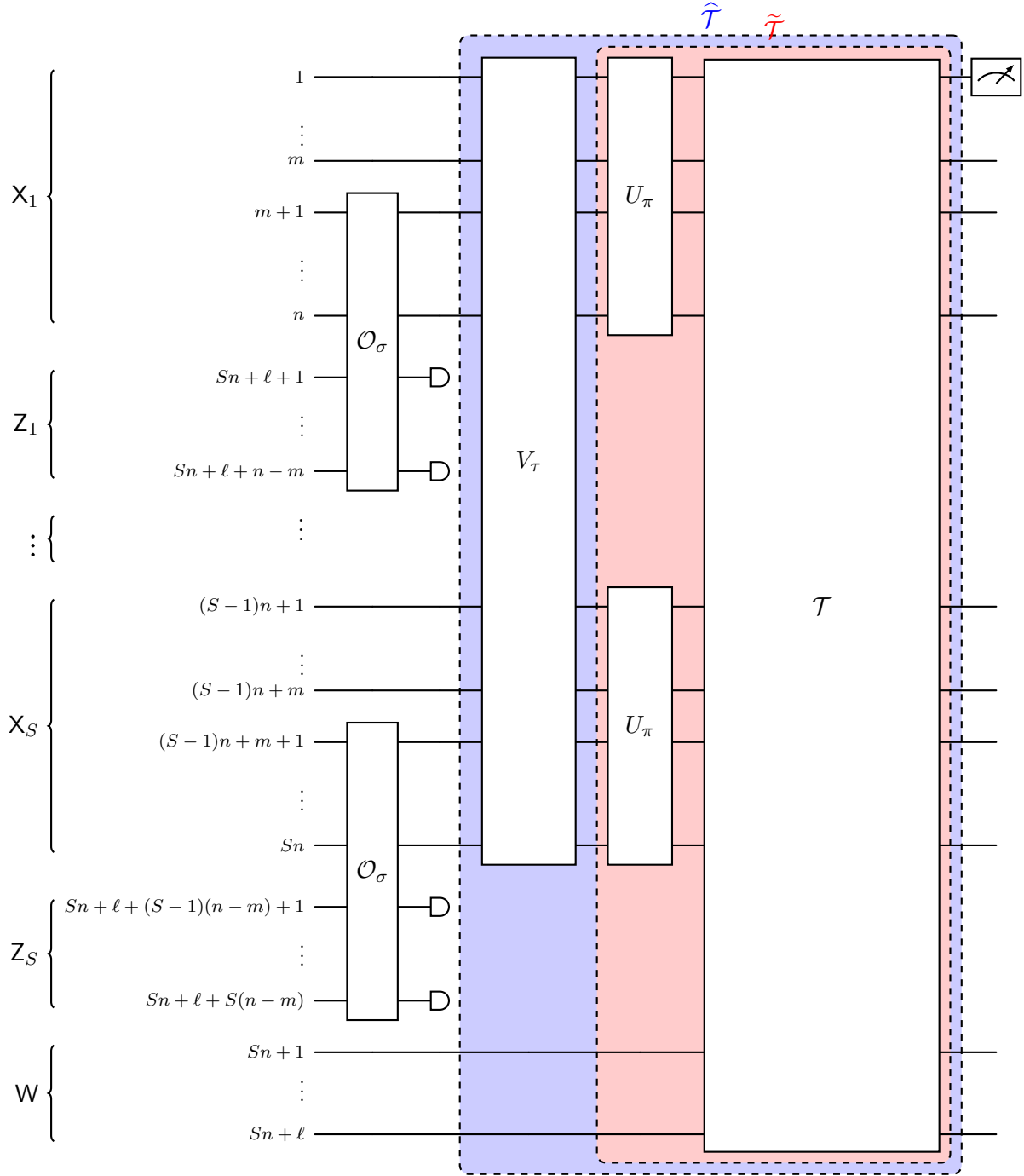

    \paragraph{Part 7: Instance encoding with tester $\overline{\mathcal{T}}$.}
    To complete our construction, we introduce extra registers of qubits: $\mathsf{Z}_1, \mathsf{Z}_2, \dots, \mathsf{Z}_{S}$, where 
    \begin{itemize}
        \item $\mathsf{Z}_{s} = \set{nS+\ell+\rbra{s-1}\rbra{n-m}+x}{x \in \sbra{n-m}}$ for $s \in \sbra{S}$. 
    \end{itemize}
    Let $\mathcal{O}_{\sigma}$ be a quantum unitary operator that prepares a purification of $\sigma$, i.e., 
    \begin{equation} \label{def:psi}
        \rbra*{\mathcal{O}_{\sigma}}_{\mathsf{X}'_s\mathsf{Z}_s} \cdot \ket{0}_{\mathsf{X}'_s\mathsf{Z}_s} = \ket{\psi}_{\mathsf{X}'_s\mathsf{Z}_s}, 
        \text{ and } \tr_{\mathsf{Z}_s}\rbra*{\ketbra{\psi}{\psi}_{\mathsf{X}'_s\mathsf{Z}_s}} = \sigma_{\mathsf{X}'_s},
    \end{equation}
    where we write $\mathsf{X}'_s$ to denote the registers $\mathsf{X}_{s,m+1}\dots\mathsf{X}_{s,n}$.
    Now we define
    \begin{equation} \label{def:barT}
        \overline{\mathcal{T}}_{\mathsf{XWZ}} = \rbra*{\widehat{\mathcal{T}}_{\mathsf{XW}} \otimes I_{\mathsf{Z}}} \cdot \rbra*{ \bigotimes_{s \in \sbra{S}} \rbra*{\mathcal{O}_\sigma}_{\mathsf{X}'_s\mathsf{Z}_s} \otimes \bigotimes_{s \in \sbra{S}} I_{\mathsf{X}_{s,1}\dots\mathsf{X}_{s,m}} \otimes I_{\mathsf{W}}},
    \end{equation}
    where the circuit implementation of $\overline{\mathcal{T}}$ is visualized in \cref{fig:view}. 
    Then, $\overline{\mathcal{T}}$ can be understood as a tester that takes $c_{r^*}$ samples of $\rho$ as input.
    For clarity, we define that
    \begin{equation} \label{def:prob-barT}
    \Pr\sbra{\overline{\mathcal{T}} \textup{ accepts } \rho} = \tr\rbra*{ \overline{\Pi}_{\mathsf{XWZ}} \cdot  \overline{\mathcal{T}}_{\mathsf{XWZ}} \rbra*{ \bigotimes_{s \in \sbra{c_{r^*}}} \overline{\rho}_{\mathsf{X}_s} \otimes \bigotimes_{s = c_{r^*}+1}^{S} \ketbra{0}{0}^{\otimes n}_{\mathsf{X}_s} \otimes \ketbra{0}{0}^{\otimes \rbra{\ell + \rbra{n-m}S}}_{\mathsf{WZ}} } \overline{\mathcal{T}}^\dag_{\mathsf{XWZ}} },
    \end{equation}
    where $\overline{\Pi}_{\mathsf{XWZ}} = \Pi_{\mathsf{XW}} \otimes I_{\mathsf{Z}}$ and $\overline{\rho} = \rho \otimes \ketbra{0}{0}^{\otimes \rbra{n-m}}$.

    Then, we can show (in \cref{lemma:prob-barT=prob-T}) that $\Pr\sbra{\overline{\mathcal{T}} \textup{ accepts } \rho} = \Pr\sbra{\mathcal{T} \textup{ accepts } U_{\pi}\rbra{\rho \otimes \sigma}U_{\pi}^\dag}$. This means that
    \begin{itemize}
        \item $\Pr\sbra{\overline{\mathcal{T}} \textup{ accepts } \rho} \geq 2/3$ if $\rho \in \mathcal{Q}_m^{\textup{yes}}$, 
        \item $\Pr\sbra{\overline{\mathcal{T}} \textup{ accepts } \rho} \leq 1/3$ if $\rho \in \mathcal{Q}_m^{\textup{no}}$.
    \end{itemize}
    That is, $\overline{\mathcal{T}}$ is a tester for $\mathcal{Q}_{m}$ with sample complexity $c_{r^*}$ and auxiliary space complexity $\rbra{2n-m}S-mc_{r^*}+\ell$. 

    \paragraph{Part 8: The sample complexity of $\overline{\mathcal{T}}$.}
    Since $\overline{\mathcal{T}}$ is a tester for $\mathcal{Q}_m$ with sample complexity $c_{r^*}$, by the definition of sample complexity, we have $c_{r^*} \geq \mathsf{S}\rbra{\mathcal{Q}_m}$.
    By \cref{eq:c-vs-T/R}, we have
    \begin{equation}
        \frac{T}{R} \geq \mathsf{S}\rbra{\mathcal{Q}_m},
    \end{equation}
    which gives
    \begin{equation}
        \mathsf{T}\rbra{\mathcal{P}_n} = T \geq R \cdot \mathsf{S}\rbra{\mathcal{Q}_m},
    \end{equation}
    thereby completing the proof.

    \subsection{Technical lemmas}

    \begin{lemma} \label{lemma:acts-trivial}
        $\mathcal{T}^\dag \Pi \mathcal{T}$ acts trivially on $\mathsf{X}_{s,x}$ for all $s \in \sbra{S} \setminus C_{r^*}$ and $x \in \pi\rbra{\sbra{m}}$, where $\pi$ is a permutation defined by \cref{def:pi}.
        That is, $\mathcal{T}^\dag \Pi \mathcal{T}$ can be written as
        \begin{equation}
            \mathcal{T}^\dag \Pi \mathcal{T} = P_{\mathsf{XW} \setminus \mathsf{X}^*} \otimes I_{\mathsf{X}^*}, 
        \end{equation}
        where $P_{\mathsf{XW} \setminus \mathsf{X}^*}$ is a projection operator that acts on all registers except those in $\mathsf{X}^*$, and $\mathsf{X}^*$ denotes the set of all registers $\mathsf{X}_{s, x}$ for $s \in \sbra{S}\setminus C_{r^*}$ and $x \in \pi\rbra{\sbra{m}}$. 
        Moreover,
        \begin{equation}
            \tr\rbra*{\mathcal{T}^\dag \Pi \mathcal{T} \rho} = \tr\rbra*{P_{\mathsf{XW} \setminus \mathsf{X}^*} \cdot \tr_{\mathsf{X}^*}\rbra*{\rho}}.
        \end{equation}
    \end{lemma}
    \begin{proof}
    Let $K \subseteq \sbra{nS+\ell}$ be the connected component of vertex $1$ in the graph $G$ defined by \cref{eq:def-G}. 
    Let $\overline{K} = \sbra{nS+\ell}\setminus K$.

    First, we show that every register in $\mathsf{X}^*$ belongs to $\overline{K}$. Recall that
    \begin{equation}
        \bigcup_{j \in \sbra{k}} A_j^{\rbra{r^*}} = \set{\pi\rbra{x}}{x \in \sbra{m}} = \pi\rbra{\sbra{m}}.
    \end{equation}
    Thus, for any $s \in \sbra{S} \setminus C_{r^*}$ and $x\in\pi\rbra{\sbra{m}}$, we have: (i) $x\in A_j^{\rbra{r^*}}$ for some $j\in\sbra{k}$; (2) by the definition of $C_{r^*}$, the vertices $\rbra{s-1}n+x$ and $1$ are not connected in $G$. Hence, $\mathsf{X}_{s,x}\subseteq \overline{K}$ for all $s \in \sbra{S} \setminus C_{r^*}$ and $x\in\pi\rbra{\sbra{m}}$, which gives $\mathsf{X}^*\subseteq \overline{K}$.

    Since $K$ and $\overline{K}$ are distinct connected components, no edge of $G$ has one endpoint in $K$ and the other endpoint in $\overline{K}$. 
    Therefore, every two-qubit gate in $\mathcal{T}$ acts either entirely on $K$ or entirely on $\overline{K}$. Since any gates acting on disjoint registers commute, the circuit $\mathcal{T}$ can be written in the form
    \begin{equation}
        \mathcal{T} = U_{K} \otimes V_{\overline{K}},
    \end{equation}
    for some unitary operators $U$ and $V$.

    On the other hand, since vertex $1$ belongs to $K$, the measurement projector has the form
    \begin{equation}
        \Pi
        =
        \rbra*{\ketbra{0}{0}_{\mathsf{X}_{1,1}}
        \otimes I_{K\setminus\{1\}}}
        \otimes I_{\overline{K}}.
    \end{equation}
    It follows that
    \begin{align}
        \mathcal{T}^{\dag}\Pi\mathcal{T}
        &=
        \rbra*{
            U_{K}^{\dag}
            \rbra*{\ketbra{0}{0}_{\mathsf{X}_{1,1}}
            \otimes I_{K\setminus\{1\}}}
            U_{K}
        }
        \otimes I_{\overline{K}} \\
        & = \rbra*{
            U_{K}^{\dag}
            \rbra*{\ketbra{0}{0}_{\mathsf{X}_{1,1}}
            \otimes I_{K\setminus\{1\}}}
            U_{K}
        }
        \otimes I_{\overline{K} \setminus \mathsf{X}^*} \otimes I_{\mathsf{X}^*}.
    \end{align}
    Hence, there exists a projection operator
    $P_{\mathsf{XW}\setminus\mathsf{X}^*}$ such that
    \begin{equation}
        \mathcal{T}^{\dag}\Pi\mathcal{T}
        =
        P_{\mathsf{XW}\setminus\mathsf{X}^*}
        \otimes I_{\mathsf{X}^*},
    \end{equation}
    where
    \begin{equation}
        P_{\mathsf{XW}\setminus\mathsf{X}^*} = \rbra*{
            U_{K}^{\dag}
            \rbra*{\ketbra{0}{0}_{\mathsf{X}_{1,1}}
            \otimes I_{K\setminus\{1\}}}
            U_{K}
        }
        \otimes I_{\overline{K} \setminus \mathsf{X}^*}.
    \end{equation}

    Finally, for any density operator $\rho$ on $\mathsf{XW}$,
    \begin{equation}
        \tr\rbra*{\mathcal{T}^{\dag}\Pi\mathcal{T}\rho}
        =
        \tr\rbra*{
            \rbra*{P_{\mathsf{XW}\setminus\mathsf{X}^*}
            \otimes I_{\mathsf{X}^*}}\rho
        }
        =
        \tr\rbra*{
            P_{\mathsf{XW}\setminus\mathsf{X}^*}
            \cdot
            \tr_{\mathsf{X}^*}\rbra*{\rho}
        },
    \end{equation}
    which completes the proof.
    \end{proof}

    \begin{lemma} \label{lemma:prob-tildeT}
        Let $\eta$ and $\widetilde{\mathcal{T}}$ be defined by \cref{def:eta,def:tildeT}. 
        Then, 
        \begin{equation}
            \tr\rbra*{\Pi \widetilde{\mathcal{T}} \rbra*{ \eta_{\mathsf{X}} \otimes \ketbra{0}{0}^{\otimes \ell}_{\mathsf{W}}} \widetilde{\mathcal{T}}^\dag} = \Pr\sbra*{\mathcal{T} \textup{ accepts } U_{\pi}\rbra{\rho \otimes \sigma}U_{\pi}^\dag}.
        \end{equation}
    \end{lemma}
    \begin{proof}
        Direct calculation shows that
        \begin{align}
            & \tr\rbra*{\Pi \widetilde{\mathcal{T}} \rbra*{ \eta_{\mathsf{X}} \otimes \ketbra{0}{0}^{\otimes \ell}_{\mathsf{W}}} \widetilde{\mathcal{T}}^\dag} \\
            & = \tr\rbra*{ \Pi \mathcal{T} \rbra*{ \bigotimes_{s \in \sbra{S}} \rbra{U_{\pi}}_{\mathsf{X}_s} \otimes I_{\mathsf{W}} } \rbra*{ \eta_{\mathsf{X}} \otimes \ketbra{0}{0}^{\otimes \ell}_{\mathsf{W}} } \rbra*{ \bigotimes_{s \in \sbra{S}} \rbra{U_{\pi}^\dag}_{\mathsf{X}_s} \otimes I_{\mathsf{W}} } \mathcal{T}^{\dag} } \\
            & = \tr\rbra*{ \Pi \mathcal{T} \rbra*{ \bigotimes_{s \in C_{r^*}} \rbra*{U_{\pi} \rbra*{\rho \otimes \sigma} U_{\pi}^\dag}_{\mathsf{X}_s} \otimes \bigotimes_{s \in \sbra{S} \setminus C_{r^*}} \rbra*{U_{\pi} \rbra*{\ketbra{0}{0}^{\otimes m} \otimes \sigma} U_{\pi}^\dag}_{\mathsf{X}_s} \otimes \ketbra{0}{0}^{\otimes \ell}_{\mathsf{W}} } \mathcal{T}^\dag }. \label{eq:prob-T-Cr}
        \end{align}
        By \cref{lemma:acts-trivial},  $\mathcal{T}^\dag \Pi \mathcal{T}$ has the form
        \begin{equation}
            \mathcal{T}^\dag \Pi \mathcal{T} = P_{\mathsf{XW} \setminus \mathsf{X}^*} \otimes I_{\mathsf{X}^*}, 
        \end{equation}
        for some projection operator $P_{\mathsf{XW} \setminus \mathsf{X}^*}$ on registers ${\mathsf{XW} \setminus \mathsf{X}^*}$, and $\mathsf{X}^*$ denotes the set of all registers $\mathsf{X}_{s, x}$ for $s \in \sbra{S}\setminus C_{r^*}$ and $x \in \pi\rbra{\sbra{m}}$, where $\pi$ is a permutation defined by \cref{def:pi}. 
        On the other hand, for every $s \in \sbra{S}$, 
        \begin{align}
            \rbra*{U_{\pi} \rbra*{\ketbra{0}{0}^{\otimes m} \otimes \sigma} U_{\pi}^\dag}_{\mathsf{X}_s}
            & = \rbra{U_{\pi}}_{\mathsf{X}_s} \rbra*{\ketbra{0}{0}^{\otimes m}_{\mathsf{X}_{s,1}\dots\mathsf{X}_{s,m}} \otimes \sigma_{\mathsf{X}_{s,m+1}\dots\mathsf{X}_{s,n}}} \rbra{U_{\pi}^\dag}_{\mathsf{X}_s} \\
            & = \ketbra{0}{0}^{\otimes m}_{\mathsf{X}_{s,\pi\rbra{1}}\dots\mathsf{X}_{s,\pi\rbra{m}}} \otimes \sigma_{\mathsf{X}_{s,\pi\rbra{m+1}}\dots\mathsf{X}_{s,\pi\rbra{n}}},
        \end{align}
        which yields
        \begin{equation}
            \bigotimes_{s \in \sbra{S} \setminus C_{r^*}} \rbra*{U_{\pi} \rbra*{\ketbra{0}{0}^{\otimes m} \otimes \sigma} U_{\pi}^\dag}_{\mathsf{X}_s}
            = \rbra*{\bigotimes_{s \in \sbra{S} \setminus C_{r^*}} \sigma_{\mathsf{X}_{s,\pi\rbra{m+1}}\dots\mathsf{X}_{s,\pi\rbra{n}}} } \otimes \ketbra{0}{0}^{\otimes m\rbra{S-c_{r^*}}}_{\mathsf{X}^*}.
        \end{equation}
        Then, 
        \begin{equation}
            \eqref{eq:prob-T-Cr} 
            = \tr\rbra*{ P_{\mathsf{XW} \setminus \mathsf{X}^*} \rbra*{ \bigotimes_{s \in C_{r^*}} \rbra*{U_{\pi} \rbra*{\rho \otimes \sigma} U_{\pi}^\dag}_{\mathsf{X}_s} \otimes \bigotimes_{s \in \sbra{S} \setminus C_{r^*}} \sigma_{\mathsf{X}_{s,\pi\rbra{m+1}}\dots\mathsf{X}_{s,\pi\rbra{n}}} \otimes \ketbra{0}{0}^{\otimes \ell}_{\mathsf{W}} } }.
        \end{equation}
        Similarly, we have
        \begin{equation}
            \eqref{eq:prob-T-pi} = \tr\rbra*{ P_{\mathsf{XW} \setminus \mathsf{X}^*} \rbra*{ \bigotimes_{s \in C_{r^*}} \rbra*{U_{\pi} \rbra*{\rho \otimes \sigma} U_{\pi}^\dag}_{\mathsf{X}_s} \otimes \bigotimes_{s \in \sbra{S} \setminus C_{r^*}} \sigma_{\mathsf{X}_{s,\pi\rbra{m+1}}\dots\mathsf{X}_{s,\pi\rbra{n}}} \otimes \ketbra{0}{0}^{\otimes \ell}_{\mathsf{W}} } }.
        \end{equation}
        Therefore, $\eqref{eq:prob-T-pi} = \eqref{eq:prob-T-Cr}$, i.e.,
        \begin{equation}
            \Pr\sbra*{\mathcal{T} \textup{ accepts } U_{\pi}\rbra{\rho \otimes \sigma}U_{\pi}^\dag} = 
            \tr\rbra*{\Pi \widetilde{\mathcal{T}} \rbra*{ \eta_{\mathsf{X}} \otimes \ketbra{0}{0}^{\otimes \ell}_{\mathsf{W}}} \widetilde{\mathcal{T}}^\dag}.
        \end{equation}
    \end{proof}

    \begin{lemma} \label{lemma:prob-tildeT=prob-hatT}
        Let $\widetilde{\mathcal{T}}$ and $\widehat{\mathcal{T}}$ be defined by \cref{def:tildeT,def:hatT}, respectively. Then,
        \begin{equation}
            \tr\rbra*{\Pi \widehat{\mathcal{T}} \rbra*{ \bigotimes_{s \in \sbra{c_{r^*}}} \rbra*{\rho \otimes \sigma}_{\mathsf{X}_s} \otimes \bigotimes_{s = c_{r^*}+1}^{S} \rbra*{\ketbra{0}{0}^{\otimes m} \otimes \sigma}_{\mathsf{X}_s} \otimes \ketbra{0}{0}^{\otimes \ell}_{ \mathsf{W}} } \widehat{\mathcal{T}}^\dag} = \tr\rbra*{\Pi \widetilde{\mathcal{T}} \rbra*{ \eta_{\mathsf{X}} \otimes \ketbra{0}{0}^{\otimes \ell}_{\mathsf{W}}} \widetilde{\mathcal{T}}^\dag}.
        \end{equation}
    \end{lemma}
    \begin{proof}
        Note that $\tau \in \Sym\rbra{S}$ is a permutation such that $\tau\rbra{s_j} = j$ for all $j \in \sbra{c_{r^*}}$, and thus $\tau\rbra{C_{r^*}} = \sbra{c_{r^*}}$ and $\tau\rbra{\sbra{S} \setminus C_{r^*}} = \sbra{S}\setminus\sbra{c_{r^*}}$. 
        Then, 
        \begin{align}
            \eta_{\mathsf{X}}
            & = \bigotimes_{s \in \tau^{-1}\rbra{\sbra{c_{r^*}}}} \rbra*{\rho \otimes \sigma}_{\mathsf{X}_s} \otimes \bigotimes_{s \in \tau^{-1}\rbra{\sbra{S}\setminus\sbra{c_{r^*}}}} \rbra*{\ketbra{0}{0}^{\otimes m} \otimes \sigma}_{\mathsf{X}_s} \\
            & = \rbra{V_\tau}_{\mathsf{X}} \cdot \rbra*{\bigotimes_{s \in \sbra{c_{r^*}}} \rbra*{\rho \otimes \sigma}_{\mathsf{X}_s} \otimes \bigotimes_{s = c_{r^*}+1}^{S} \rbra*{\ketbra{0}{0}^{\otimes m} \otimes \sigma}_{\mathsf{X}_s}} \cdot \rbra{V_\tau^\dag}_{\mathsf{X}},
        \end{align}
        which yields the proof. 
    \end{proof}

    \begin{lemma} \label{lemma:prob-barT=prob-T}
        Let $\mathcal{T}$ be a tester satisfying the condition in \cref{def:Tester}. 
        Let $\overline{\mathcal{T}}$ be defined by \cref{def:barT}. Then,
        $\Pr\sbra{\overline{\mathcal{T}} \textup{ accepts } \rho} = \Pr\sbra{\mathcal{T} \textup{ accepts } U_{\pi}\rbra{\rho \otimes \sigma}U_{\pi}^\dag}$. 
    \end{lemma}
    \begin{proof}
        By \cref{lemma:prob-tildeT,lemma:prob-tildeT=prob-hatT}, we have 
        \begin{equation}
            \Pr\sbra{\mathcal{T} \textup{ accepts } U_{\pi}\rbra{\rho \otimes \sigma}U_{\pi}^\dag} = \tr\rbra*{\Pi \widehat{\mathcal{T}} \rbra*{ \bigotimes_{s \in \sbra{c_{r^*}}} \rbra*{\rho \otimes \sigma}_{\mathsf{X}_s} \otimes \bigotimes_{s = c_{r^*}+1}^{S} \rbra*{\ketbra{0}{0}^{\otimes m} \otimes \sigma}_{\mathsf{X}_s} \otimes \ketbra{0}{0}^{\otimes \ell}_{ \mathsf{W}} } \widehat{\mathcal{T}}^\dag}.
        \end{equation}
        By \cref{def:prob-barT}, we have 
        \begin{equation}
            \Pr\sbra{\overline{\mathcal{T}} \textup{ accepts } \rho} = \tr\rbra*{ \Pi \tr_{\mathsf{Z}}\rbra*{\overline{\mathcal{T}}_{\mathsf{XWZ}} \rbra*{ \bigotimes_{s \in \sbra{c_{r^*}}} \overline{\rho}_{\mathsf{X}_s} \otimes \bigotimes_{s = c_{r^*}+1}^{S} \ketbra{0}{0}^{\otimes n}_{\mathsf{X}_s} \otimes \ketbra{0}{0}^{\otimes \rbra{\ell + \rbra{n-m}S}}_{\mathsf{WZ}} } \overline{\mathcal{T}}^\dag_{\mathsf{XWZ}}} },
        \end{equation}
        which first traces out the system $\mathsf{Z}$. 
        By \cref{def:barT}, $\overline{\mathcal{T}}_{\mathsf{XWZ}}$ first applies the state-preparation circuit $\mathcal{O}_{\sigma}^{\otimes S}$ for $\sigma$, and then performs $\widehat{\mathcal{T}}_{\mathsf{XW}}$. 
        Note that
        \begin{align}
             & \tr_{\mathsf{Z}}\rbra*{\overline{\mathcal{T}}_{\mathsf{XWZ}} \rbra*{ \bigotimes_{s \in \sbra{c_{r^*}}} \overline{\rho}_{\mathsf{X}_s} \otimes \bigotimes_{s = c_{r^*}+1}^{S} \ketbra{0}{0}^{\otimes n}_{\mathsf{X}_s} \otimes \ketbra{0}{0}^{\otimes \rbra{\ell + \rbra{n-m}S}}_{\mathsf{WZ}} } \overline{\mathcal{T}}^\dag_{\mathsf{XWZ}}} \\
             & \quad = \widehat{\mathcal{T}}_{\mathsf{XW}} \rbra*{\bigotimes_{s \in \sbra{c_{r^*}}} \rho_{\overline{\mathsf{X}}_s} \otimes \bigotimes_{s = c_{r^*}+1}^{S} \ketbra{0}{0}^{\otimes m}_{\overline{\mathsf{X}}_s} \otimes \bigotimes_{s \in \sbra{S}} \tr_{\mathsf{Z}} \rbra*{\ketbra{\psi}{\psi}_{\mathsf{X}'_s\mathsf{Z}_s}} \otimes \ketbra{0}{0}^{\otimes \ell}_{ \mathsf{W}} } \widehat{\mathcal{T}}_{\mathsf{XW}}^\dag \\
             & \quad = \widehat{\mathcal{T}}_{\mathsf{XW}} \rbra*{\bigotimes_{s \in \sbra{c_{r^*}}} \rho_{\overline{\mathsf{X}}_s} \otimes \bigotimes_{s = c_{r^*}+1}^{S} \ketbra{0}{0}^{\otimes m}_{\overline{\mathsf{X}}_s} \otimes \sigma^{\otimes S}_{\mathsf{X}'} \otimes \ketbra{0}{0}^{\otimes \ell}_{ \mathsf{W}} } \widehat{\mathcal{T}}_{\mathsf{XW}}^\dag \\
             & \quad = \widehat{\mathcal{T}}_{\mathsf{XW}} \rbra*{ \bigotimes_{s \in \sbra{c_{r^*}}} \rbra*{\rho \otimes \sigma}_{\mathsf{X}_s} \otimes \bigotimes_{s = c_{r^*}+1}^{S} \rbra*{\ketbra{0}{0}^{\otimes m} \otimes \sigma}_{\mathsf{X}_s} \otimes \ketbra{0}{0}^{\otimes \ell}_{ \mathsf{W}} } \widehat{\mathcal{T}}_{\mathsf{XW}}^\dag,
        \end{align}
        where $\ket{\psi}$ is defined by \cref{def:psi} and $\overline{\mathsf{X}}_s$ denotes the registers $\mathsf{X}_{s, 1} \dots \mathsf{X}_{s, m}$. 
        
        Therefore, we conclude that $\Pr\sbra{\overline{\mathcal{T}} \textup{ accepts } \rho} = \Pr\sbra{\mathcal{T} \textup{ accepts } U_{\pi}\rbra{\rho \otimes \sigma}U_{\pi}^\dag}$. 
    \end{proof}

    \section{Applications}

    In this section, we present the applications of our main theorem. 

    \subsection{Purity estimation} \label{sec:purity}

    We first study the problem of purity estimation. 
    The definitions of its variants are given as follows. 

    \begin{definition} [Purity estimation]
        Let $n \geq 1$ be an integer and $0 \leq a < b \leq 1$. 
        We define $\textsc{Purity}\sbra{n, a, b}$ to be an $n$-qubit property such that
        \begin{align}
            \textsc{Purity}\sbra{n, a, b}^{\textup{yes}} & = \set{\rho \in \mathcal{D}\rbra{\mathcal{H}_2^{\otimes n}}}{\tr\rbra{\rho^2} \geq b}, \\
            \textsc{Purity}\sbra{n, a, b}^{\textup{no}} & = \set{\rho \in \mathcal{D}\rbra{\mathcal{H}_2^{\otimes n}}}{\tr\rbra{\rho^2} \leq a}. 
        \end{align}
    \end{definition}

    First of all, we note that purity is permutation-invariant. 

    \begin{fact} \label{fact:purity-perm-invariant}
        $\textsc{Purity}\sbra{n, a, b}$ is permutation-invariant. 
    \end{fact}
    \begin{proof}
        Actually, it is unitarily invariant, and thus permutation-invariant.
    \end{proof}

    Then, we can establish a general lower bound for $\textsc{Purity}\sbra{n, a, b}$.

    \begin{lemma} \label{lemma:general-purity}
        $\mathsf{T}\rbra{\textsc{Purity}\sbra{n, a, b}} \geq n \cdot \mathsf{S}\rbra{\textsc{Purity}\sbra{1, a, b}}$.
    \end{lemma}
    \begin{proof}
        It can be verified that $\textsc{Purity}\sbra{1, a, b} \xhookrightarrow[]{\ketbra{0}{0}^{\otimes \rbra{n-1}}} \textsc{Purity}\sbra{n, a, b}$.
        Note that $\textsc{Purity}\sbra{n, a, b}$ is permutation-invariant (by \cref{fact:purity-perm-invariant}). 
        Then, the proof is completed by applying \cref{thm:main-sp}. 
    \end{proof}

    To establish tight lower bounds on the time complexity of purity estimation, we mention the known lower bounds on its sample complexity in certain cases. 

    \begin{lemma} [Sample lower bounds for purity estimation, adapted from \cite{SW22,CWLY23,CWZ24,GHYZ24}] \label{lemma:samp-purity}
        We have 
        $\mathsf{S}\rbra{\textsc{Purity}\sbra{1, 1 - \varepsilon, 1}} = \Omega\rbra{1/\varepsilon}$ for $0 < \varepsilon \leq \frac{1}{2}$ and $\mathsf{S}\rbra{\textsc{Purity}\sbra{1, \frac{5}{9} - \varepsilon, \frac{5}{9}}} = \Omega\rbra{1/\varepsilon^2}$ for $0 < \varepsilon < \frac{1}{36}$. 
    \end{lemma}
    \begin{proof}
        The proof of each case is well-known in the literature. 
        Here, we have to emphasize that their hard instances can involve only $1$-qubit states.
        For $\textsc{Purity}\sbra{1, 1 - \varepsilon, 1}$, the hard instance is:
        \begin{equation}
            \rho^{\textup{yes}} = \ketbra{0}{0}, \qquad \rho^{\textup{no}} = \rbra{1-\delta}\ketbra{0}{0} + \delta\ketbra{1}{1}. 
        \end{equation}
        For $\textsc{Purity}\sbra{1, \frac{5}{9} - \varepsilon, \frac{5}{9}}$, the hard instance is:
        \begin{equation}
            \rho^{\textup{yes}} = \frac{2}{3}\ketbra{0}{0} + \frac{1}{3}\ketbra{1}{1}, \qquad \rho^{\textup{no}} = \rbra*{\frac{2}{3}-\delta}\ketbra{0}{0} + \rbra*{\frac{1}{3}+\delta}\ketbra{1}{1}. 
        \end{equation}
        In both cases, $\delta = \Theta\rbra{\varepsilon}$.
    \end{proof}

    Now we are ready to establish tight lower bounds on the time complexity of purity estimation. 

    \begin{theorem} [Quantum time lower bounds for purity estimation] \label{thm:lb-purity}
        We have $\mathsf{T}\rbra{\textsc{Purity}\sbra{n, 1 - \varepsilon, 1}} \geq \Omega\rbra{n/\varepsilon}$ for $0 < \varepsilon \leq \frac{1}{2}$ and $\mathsf{T}\rbra{\textsc{Purity}\sbra{n, \frac{5}{9} - \varepsilon, \frac{5}{9}}} \geq \Omega\rbra{n/\varepsilon^2}$ for $0 < \varepsilon < \frac{1}{36}$. 
    \end{theorem}
    \begin{proof}
        This is obtained by applying \cref{lemma:general-purity} with \cref{lemma:samp-purity}. 
    \end{proof}

    \subsection{Productness testing}

    In this section, we consider productness testing. 
    The formal definition is given as follows. 

    \begin{definition} [Productness testing]
        Let $n, m \geq 1$ be integers and $0 < \varepsilon < 1$. 
        We define $\textsc{Productness}\sbra{n, m, \varepsilon}$ to be an $nm$-qubit property such that 
        \begin{equation}
        \textsc{Productness}\sbra{n, m, \varepsilon}^{\textup{yes}} = \set{\ket{\psi} = \bigotimes_{j \in \sbra{n}} \ket{\psi_j}}{ \ket{\psi_j} \in \mathcal{H}_2^{\otimes m} \textup{ for } j \in \sbra{n} },
    \end{equation}
    \begin{equation}
        \textsc{Productness}\sbra{n, m, \varepsilon}^{\textup{no}} = \set{ \ket{\psi} }{ \mathrm{T}\rbra*{\ket{\psi}, \bigotimes_{j \in \sbra{n}} \ket{\psi_j}} \geq \varepsilon \textup{ for any } \ket{\psi_j} \in \mathcal{H}_2^{\otimes m}, j \in \sbra{n} }.
    \end{equation}
    \end{definition}
    
    Although a matching quantum sample complexity lower bound $\Omega\rbra{1/\varepsilon^2}$ for productness testing was already noted in \cite{SW22,CWZ24}, we re-derive it here for completeness. 
    \begin{lemma} [Sample lower bound for productness testing, adapted from \cite{SW22,CWZ24}] \label{lemma:s-product}
        For every $0 < \varepsilon \leq \frac{1}{2}$, we have $\mathsf{S}\rbra{\textsc{Productness}\sbra{2, 1, \varepsilon}} \geq \Omega\rbra{{1}/{\varepsilon^2}}$. 
    \end{lemma}
    \begin{proof}
        According to the proof of \cite[Theorem 4.13]{CWZ24}, we have
        \begin{equation} \label{eq:product-purity}
            \mathsf{S}\rbra{\textsc{Productness}\sbra{2, 1, \varepsilon}} \geq \mathsf{S}\rbra{\textsc{Purity}\sbra{1, 1-\varepsilon^2, 1}}.
        \end{equation}
        Then, by \cref{lemma:samp-purity}, we have $\eqref{eq:product-purity} \geq \Omega\rbra{1/\varepsilon^2}$. 
    \end{proof}

    Now we are ready to establish a tight lower bound on the quantum time complexity of productness testing. 

    \begin{theorem} [Quantum time lower bounds for productness testing] \label{thm:prod-test}
        For $n \geq 1$, $m \geq 1$, and $0 < \varepsilon \leq \frac{1}{2}$, we have $\mathsf{T}\rbra{\textsc{Productness}\sbra{2n, m, \varepsilon}} \geq \Omega\rbra{{nm}/{\varepsilon^2}}$.
    \end{theorem}
    \begin{proof}
        $\textsc{Productness}\sbra{2n, m, \varepsilon}$ is a property of $2nm$-qubit pure quantum states. 
        Here, we number the $2nm$ qubits from $1$ to $2nm$ as follows:
        \begin{itemize}
            \item The $k$-th qubit of the $j$-th part is numbered $\rbra{k - 1}2n+j$ for $j \in \sbra{2n}$ and $k \in \sbra{m}$. 
        \end{itemize}
        Let 
        \begin{equation}
            \mathcal{G} = \Sym\rbra{A_0} \times \Sym\rbra{A_1} \subseteq \Sym\rbra{2nm},
        \end{equation}
        where
        \begin{equation}
            A_{b} = \set{\rbra{k - 1}2n+j}{j \in \sbra{2n} \textup{ with } j \equiv b \pmod 2, k \in \sbra{m}}
        \end{equation}
        for $b \in \cbra{0, 1}$.
        It can be shown that 
        \begin{equation}
            \textsc{Productness}\sbra{2, 1, \varepsilon} \xhookrightarrow[\mathcal{G}]{\ket{0}^{\otimes \rbra{2nm-2}}} \textsc{Productness}\sbra{2n, m, \varepsilon}.
        \end{equation}
        Then, by \cref{thm:main}, we have
        \begin{equation} \label{eq:t-s-product}
            \mathsf{T}\rbra{\textsc{Productness}\sbra{2n, m, \varepsilon}} \geq R \cdot \mathsf{S}\rbra{\textsc{Productness}\sbra{2, 1, \varepsilon}},
        \end{equation}
        where
        \begin{equation}
            R = \min_{b \in \cbra{0, 1}} \floor*{\frac{\abs{A_b}}{\abs{A_b \cap \sbra{2}}}} = nm.
        \end{equation}
        By \cref{lemma:s-product}, we have $\eqref{eq:t-s-product} \geq \Omega\rbra{{nm}/{\varepsilon^2}}$.
    \end{proof}

    \subsection{Inner product estimation}

    In this section, we consider the problem of inner product estimation. 
    The formal definition is given as follows.

    \begin{definition} [Inner product estimation]
        Let $n \geq 1$ be an integer and $0 \leq a < b \leq 1$. 
        We define $\textsc{InnerProduct}\sbra{n, a, b}$ to be a $2n$-qubit property such that
        \begin{align}
            \textsc{InnerProduct}\sbra{n, a, b}^{\textup{yes}} & = \set{ \rho \otimes \sigma }{\rho, \sigma \in \mathcal{D}\rbra{\mathcal{H}_2^{\otimes n}}, \tr\rbra{\rho\sigma} \geq b}, \\
            \textsc{InnerProduct}\sbra{n, a, b}^{\textup{no}} & = \set{ \rho \otimes \sigma }{\rho, \sigma \in \mathcal{D}\rbra{\mathcal{H}_2^{\otimes n}}, \tr\rbra{\rho\sigma} \leq a}. 
        \end{align}
    \end{definition}

    Note that $\textsc{InnerProduct}\sbra{n, a, b}$ is not unitarily invariant but it is still permutation-invariant. 

    \begin{fact} \label{fact:inner-prod}
        $\textsc{InnerProduct}\sbra{n, a, b}$ is permutation-invariant (from the perspective of $4$-dimensional qudits). 
    \end{fact}
    \begin{proof}
        Let $\pi \in \Sym\rbra{n}$ and $U_{\pi}$ defined by \cref{eq:def-U-pi}. 
        Then, for $X \in \cbra{\textup{yes}, \textup{no}}$, we have $\rho \otimes \sigma \in \textsc{InnerProduct}\sbra{n, a, b}^X$ if and only if $U_\pi \rho U_\pi^\dag \otimes U_\pi \sigma U_\pi^\dag \in \textsc{InnerProduct}\sbra{n, a, b}^X$. 
        If we regard the $j$-th qubit of $\rho$ and the $j$-th qubit of $\sigma$ as a whole as the $j$-th ($4$-dimensional) qudit of $\rho \otimes \sigma$, then  $\textsc{InnerProduct}\sbra{n, a, b}$ is permutation-invariant from this perspective of $4$-dimensional qudits.
    \end{proof}

    Then, we can establish a general lower bound for $\textsc{InnerProduct}\sbra{n, a, b}$.

    \begin{lemma} \label{lemma:inner-prod}
        $\mathsf{T}\rbra{\textsc{InnerProduct}\sbra{n, a, b}} \geq n \cdot \mathsf{S}\rbra{\textsc{InnerProduct}\sbra{1, a, b}}$.
    \end{lemma}
    \begin{proof}
        It can be verified that $\textsc{InnerProduct}\sbra{1, a, b} \xhookrightarrow[]{\ketbra{0}{0}^{\otimes \rbra{n-1}} \otimes \ketbra{0}{0}^{\otimes \rbra{n-1}}} \textsc{InnerProduct}\sbra{n, a, b}$.
        Note that $\textsc{InnerProduct}\sbra{n, a, b}$ is permutation-invariant (by \cref{fact:inner-prod}). 
        Then, the proof is completed by applying the $4$-dimensional qudit version of \cref{thm:main-sp} (as noted in \cref{remark:qudit}). 
    \end{proof}

    To show a lower bound on the quantum time complexity of inner product estimation, we need the following sample lower bound. 

    \begin{lemma} [Sample lower bounds for inner product estimation, adapted from {\cite[Lemma 13 in the full version]{ALL22}}] \label{lemma:samp-inner-prod}
        For every $0 < \varepsilon < \frac{1}{2}$, we have
        $\mathsf{S}\rbra{\textsc{InnerProduct}\sbra{1, \frac{1}{2}-\varepsilon, \frac{1}{2}+\varepsilon}} = \Omega\rbra{1/\varepsilon^2}$. 
    \end{lemma}
    \begin{proof}
        The hard instance is of the form $\ketbra{\psi_{\pm}}{\psi_{\pm}} \otimes \ketbra{0}{0}$, where
        \begin{equation}
            \ket{\psi_{\pm}} = \sqrt{\frac{1}{2}\pm\varepsilon} \ket{0} + \sqrt{\frac{1}{2}\mp\varepsilon} \ket{1}.
        \end{equation}
        It can be verified that $\ketbra{\psi_{+}}{\psi_{+}} \otimes \ketbra{0}{0} \in \textsc{InnerProduct}\sbra{1, \frac{1}{2}-\varepsilon, \frac{1}{2}+\varepsilon}^{\textup{yes}}$ and $\ketbra{\psi_{-}}{\psi_{-}} \otimes \ketbra{0}{0} \in \textsc{InnerProduct}\sbra{1, \frac{1}{2}-\varepsilon, \frac{1}{2}+\varepsilon}^{\textup{no}}$. 
        Note that the hard instances are pure states. 
    \end{proof}

    Now we are ready to establish a tight lower bound on the quantum time complexity of inner product estimation. 

    \begin{theorem} [Quantum time lower bounds for inner product estimation] \label{thm:lb-inner-prod}
        For every $0 < \varepsilon < \frac{1}{2}$, we have
        $\mathsf{T}\rbra{\textsc{InnerProduct}\sbra{n, \frac{1}{2}-\varepsilon, \frac{1}{2}+\varepsilon}} = \Omega\rbra{n/\varepsilon^2}$. 
    \end{theorem}

    Although the lower bound in \cref{thm:lb-purity} already implies the same lower bound for inner product estimation as given in \cref{thm:lb-inner-prod}, \cref{thm:lb-inner-prod} means that inner product estimation is hard even for pure states (which is not implied by \cref{thm:lb-purity}). 
    In addition, the proof of \cref{thm:lb-inner-prod} requires the qudit version of our main theorem, which also shows the extensibility of our results. 

    \begin{proof} [Proof of \cref{thm:lb-inner-prod}]
        This is immediately obtained by applying \cref{lemma:inner-prod} (based on the $4$-dimensional qudit version of \cref{thm:main-sp}) with \cref{lemma:samp-inner-prod}. 
    \end{proof}

    \subsection{Power trace estimation}

    In this section, we consider the problem of estimating $\tr\rbra{\rho^k}$ for large $k \geq 2$, which is a generalization of purity estimation. 
    It is noted that for non-integer $k$, estimating $\tr\rbra{\rho^k}$ has also been studied as a key step for R\'enyi/Tsallis entropy estimation \cite{AISW20,SH21,WGL+24,WZL24,LW25,CW25,CLW26}. 
    The formal definition is given as follows. 

    \begin{definition} [Power trace estimation]
        Let $n \geq 1$ and $k \geq 2$ be integers and $0 \leq a < b \leq 1$. 
        We define $\textsc{PowerTrace}\sbra{n, k, a, b}$ to be an $n$-qubit property such that
        \begin{align}
            \textsc{PowerTrace}\sbra{n, k, a, b}^{\textup{yes}} & = \set{\rho \in \mathcal{D}\rbra{\mathcal{H}_2^{\otimes n}}}{\tr\rbra{\rho^k} \geq b}, \\
            \textsc{PowerTrace}\sbra{n, k, a, b}^{\textup{no}} & = \set{\rho \in \mathcal{D}\rbra{\mathcal{H}_2^{\otimes n}}}{\tr\rbra{\rho^k} \leq a}. 
        \end{align}
    \end{definition}

    Similar to \cref{sec:purity}, we have the following properties of power traces. 

    \begin{fact} \label{fact:power-trace}
        $\textsc{PowerTrace}\sbra{n, k, a, b}$ is permutation-invariant. 
    \end{fact}

    \begin{lemma} \label{lemma:power-trace}
        $\mathsf{T}\rbra{\textsc{PowerTrace}\sbra{n, k, a, b}} \geq n \cdot \mathsf{S}\rbra{\textsc{PowerTrace}\sbra{1, k, a, b}}$.
    \end{lemma}

    \begin{proof}
    It can be verified that $\textsc{PowerTrace}\sbra{1,k,a,b}
        \xhookrightarrow[]{\ketbra{0}{0}^{\otimes(n-1)}}
        \textsc{PowerTrace}\sbra{n,k,a,b}$, since for any $\rho$,
    \begin{equation}
        \tr\rbra*{\rbra*{\rho\otimes\ketbra{0}{0}^{\otimes(n-1)}}^k}
        =
        \tr\rbra*{\rho^k}
        \cdot
        \tr\rbra*{\rbra*{\ketbra{0}{0}^{\otimes(n-1)}}^k}
        =
        \tr\rbra*{\rho^k}.
    \end{equation}
    Note that $\textsc{PowerTrace}\sbra{n, k, a, b}$ is permutation-invariant (by \cref{fact:power-trace}). 
    Then, the proof is completed by applying \cref{thm:main-sp}.
\end{proof}

    To establish tight lower bounds on the quantum time complexity of power trace estimation, we mention the known lower bounds on its sample complexity in certain cases. 

    \begin{lemma} [Sample lower bounds for power trace estimation, adapted from \cite{CWYZ25}] \label{lemma:samp-power}
        Let $k \geq 2$ be an integer.
        For every sufficiently small $\varepsilon > 0$, there exists a real number $0 < a < 1$ with $a = \Theta\rbra{1}$, we have $\mathsf{S}\rbra{\textsc{PowerTrace}\sbra{1, k, a, a+\varepsilon}} \geq \Omega\rbra{k/\varepsilon^2}$. 
    \end{lemma}
    \begin{proof}
        The hard instance is of the form 
        \begin{equation}
            \rho_{\pm} = \rbra*{1 - \frac{1}{k} \pm \frac{\delta}{k}} \ketbra{0}{0} + \rbra*{\frac{1}{k} \mp \frac{\delta}{k}} \ketbra{1}{1},
        \end{equation}
        where $\delta = \Theta\rbra{\varepsilon}$. 
        It can be verified that $\tr\rbra{\rho_+^k} - \tr\rbra{\rho_-^k} \geq \Omega\rbra{\delta}$ and $\tr\rbra{\rho_\pm^k} = \Theta\rbra{1}$. 
        Therefore, there exists a real number $a = \Theta\rbra{1}$ such that $\rho_+ \in \textsc{PowerTrace}\sbra{1, k, a, a+\varepsilon}^{\textup{yes}}$ and $\rho_- \in \textsc{PowerTrace}\sbra{1, k, a, a+\varepsilon}^{\textup{no}}$. 
        To complete the proof, we only need to note that $1 - \mathrm{F}\rbra{\rho_+, \rho_-} \leq O\rbra{\varepsilon^2/k}$. 
        Then, by the Helstrom-Holevo bound \cite{Hel67,Hol73}, $\mathsf{S}\rbra{\textsc{PowerTrace}\sbra{1, k, a, a+\varepsilon}} \geq \Omega\rbra{1/\rbra{1-\mathrm{F}\rbra{\rho_+, \rho_-}}} \geq \Omega\rbra{k/\varepsilon^2}$. 
    \end{proof}

    Now we are ready to establish a tight lower bound on the quantum time complexity of power trace estimation. 

    \begin{theorem} [Quantum time lower bounds for power trace estimation] \label{thm:pow-tr-est}
        Let $n \geq 1$ and $k \geq 2$ be integers. 
        Then, for sufficiently small $\varepsilon > 0$, there exists a real number $a = \Theta\rbra{1}$ such that $\mathsf{T}\rbra{\textsc{PowerTrace}\sbra{n, k, a, a+\varepsilon}} \geq \Omega\rbra{nk/\varepsilon^2}$. 
    \end{theorem}

    \begin{proof}
        This is obtained by applying \cref{lemma:power-trace} with \cref{lemma:samp-power}. 
    \end{proof}

    \subsection{Trace distance estimation}

    In this section, we consider the problem of trace distance estimation between pure states. 
    The formal definition is given as follows.

    \begin{definition} [Pure-state trace distance estimation]
        Let $n \geq 1$ be an integer and $0 \leq a < b \leq 1$. 
        We define $\textsc{PureTD}\sbra{n, a, b}$ to be an $n$-qubit property such that
        \begin{align}
            \textsc{PureTD}\sbra{n, a, b}^{\textup{yes}} & = \set{\ket{\psi} \in \mathcal{H}_2^{\otimes n}}{\mathrm{T}\rbra{\ket{\psi}, \ket{0}} \geq b}, \\
            \textsc{PureTD}\sbra{n, a, b}^{\textup{no}} & = \set{\ket{\psi} \in \mathcal{H}_2^{\otimes n}}{\mathrm{T}\rbra{\ket{\psi}, \ket{0}} \leq a}. 
        \end{align}
    \end{definition}

    Note that $\textsc{PureTD}\sbra{n, a, b}$ is not unitarily invariant but permutation-invariant. 

    \begin{fact} \label{fact:puretd}
        $\textsc{PureTD}\sbra{n, a, b}$ is permutation-invariant. 
    \end{fact}

    Then, we can establish a general lower bound for $\textsc{PureTD}\sbra{n, a, b}$.

    \begin{lemma} \label{lemma:puretd-general-lb}
        $\mathsf{T}\rbra{\textsc{PureTD}\sbra{n, a, b}} \geq n \cdot \mathsf{S}\rbra{\textsc{PureTD}\sbra{1, a, b}}$.
    \end{lemma}
    \begin{proof}
        It can be verified that $\textsc{PureTD}\sbra{1, a, b} \xhookrightarrow[]{\ketbra{0}{0}^{\otimes \rbra{n-1}}} \textsc{PureTD}\sbra{n, a, b}$.
        Note that $\textsc{PureTD}\sbra{n, a, b}$ is permutation-invariant (by \cref{fact:puretd}). 
        Then, the proof is completed by applying \cref{thm:main-sp}. 
    \end{proof}

    To establish tight lower bounds on the quantum time complexity of pure-state trace distance estimation, we mention the known lower bounds on its sample complexity in certain cases. 

    \begin{lemma} [Sample lower bounds for pure-state trace distance estimation, adapted from \cite{Wan24,WZ24c}] \label{lemma:samp-puretd}
        For every $0 < \varepsilon < 1$, we have $\mathsf{S}\rbra{\textsc{PureTD}\sbra{1, 0, \varepsilon}} = \Omega\rbra{1/\varepsilon^2}$.
    \end{lemma}
    \begin{proof}
        The hard instance is of the form 
        \begin{equation}
            \ket{\psi^{\textup{yes}}} = \sqrt{1 - \varepsilon^2} \ket{0} + \varepsilon \ket{1}, \qquad \ket{\psi^{\textup{no}}} = \ket{0}.
        \end{equation}
        It can be verified that $\ket{\psi^{\textup{yes}}} \in \textsc{PureTD}\sbra{1, 0, \varepsilon}^{\textup{yes}}$ and $\ket{\psi^{\textup{no}}} \in \textsc{PureTD}\sbra{1, 0, \varepsilon}^{\textup{no}}$. 
        The proof is completed by the Helstrom--Holevo bound \cite{Hel67,Hol73}. 
    \end{proof}

    Now we are ready to establish a tight lower bound on the quantum time complexity of pure-state trace distance estimation. 

    \begin{theorem} [Quantum time lower bounds for pure-state trace distance estimation] \label{thm:time-puretd}
        For every $0 < \varepsilon < 1$, we have $\mathsf{T}\rbra{\textsc{PureTD}\sbra{n, 0, \varepsilon}} \geq \Omega\rbra{n/\varepsilon^2}$. 
    \end{theorem}
    \begin{proof}
        This is obtained by applying \cref{lemma:puretd-general-lb} with \cref{lemma:samp-puretd}. 
    \end{proof}

    \cref{thm:time-puretd} will serve as a starting point for proving the quantum time lower bounds for the samplizer in \cref{sec:samplizer} and then for the LMR protocol in \cref{sec:lmr}. 

    \subsection{Samplizer} \label{sec:samplizer}

    In this section, we consider the algorithmic tool, the samplizer \cite{WZ24}.
    For simplicity, here we only consider the pure-state samplizer (as defined in \cite{WZ24c}). 

    \begin{definition} [Pure-state samplizer, simplified {\cite[Definition 6.2]{WZ24c}}] \label{def:samplizer}
        An ($n$-qubit) pure-state samplizer, denoted as $\mathsf{Samplize}^{\mathsf{pure}}_*\ave{*}$, is a converter from a quantum circuit family to a quantum channel family such that:
        for any $\delta > 0$, quantum circuit family $\mathcal{A}^U$ with query access to an $n$-qubit unitary operator $U$, and an $n$-qubit pure state $\ket{\psi}$, 
        \begin{equation}
            \Abs*{ \mathsf{Samplize}^{\mathsf{pure}}_\delta\ave{\mathcal{A}^U} \sbra{\ket{\psi}} - \mathcal{A}^{R_{\psi}} }_{\diamond} \leq \delta,
        \end{equation}
        where $R_{\psi} = I - 2\ketbra{\psi}{\psi}$ is the reflection operator about $\ket{\psi}$. 
        The sample complexity of $\mathsf{Samplize}^{\mathsf{pure}}_*\ave{*}$ is a function $S\rbra{Q, \delta}$ such that if $\mathcal{A}^U$ uses $Q$ queries to $U$, then $\mathsf{Samplize}^{\mathsf{pure}}_\delta\ave{\mathcal{A}^U} \sbra{\ket{\psi}}$ uses (at most) $S\rbra{Q, \delta}$ samples of $\ket{\psi}$. 
        Similarly, the time complexity of $\mathsf{Samplize}^{\mathsf{pure}}_*\ave{*}$ is a function $T\rbra{Q, \delta}$ such that if $\mathcal{A}^U$ uses $Q$ queries to $U$, then $\mathsf{Samplize}^{\mathsf{pure}}_\delta\ave{\mathcal{A}^U} \sbra{\ket{\psi}}$ uses (at most) $T\rbra{Q, \delta}$ (additional) two-qubit gates. 
    \end{definition}

    An efficient implementation of pure-state samplizer was given in \cite{WZ24c}. 

    \begin{theorem} [An efficient pure-state samplizer, {\cite[Theorem 6.3]{WZ24c}}] \label{thm:pure-samplizer}
        There is an implementation of $n$-qubit pure-state samplizer with sample complexity $O\rbra{Q^2/\delta}$ and time complexity $O\rbra{nQ^2/\delta}$. 
    \end{theorem}

    The goal of this section is to show the time-optimality of the implementation in \cref{thm:pure-samplizer}.
    To this end, we need the quantum query algorithm for pure-state trace distance estimation, given query access to the reflection operators about the input pure states, in \cite{WZ24c}. 

    \begin{lemma} [Quantum subroutine for pure-state trace distance estimation, adapted from {\cite[Corollary 5.3]{WZ24c}}] \label{lemma:subroutine-td}
        Given query access to the reflection operator $R_{\psi} = I - 2\ketbra{\psi}{\psi}$ about an $n$-qubit pure state $\ket{\psi}$, there is a quantum query algorithm $\mathcal{A}^{R_{\psi}}$ that estimates the trace distance $\mathrm{T}\rbra{\ket{\psi}, \ket{0}}$ to within additive error $\varepsilon$ with success probability $\geq 0.99$ using $O\rbra{1/\varepsilon}$ queries to $R_{\psi}$ and $O\rbra{n/\varepsilon}$ two-qubit gates. 
    \end{lemma}

    Now we are ready to prove the quantum time lower bounds for the pure-state samplizer. 

    \begin{theorem} [Time-optimality of samplizer] \label{thm:time-samplizer}
        Any implementation of $n$-qubit pure-state samplizer requires quantum time complexity $\Omega\rbra{nQ^2/\delta}$. 
    \end{theorem}
    \begin{proof}
        Suppose that there is an implementation of $n$-qubit pure-state samplizer
with sample complexity $S(Q,\delta)$ and time complexity $T(Q,\delta)$ as
in \cref{def:samplizer}. Applying this implementation to the quantum algorithm in \cref{lemma:subroutine-td}, with samplization precision
$0.01$, gives a quantum algorithm that estimates
$\mathrm{T}(\ket{\psi},\ket{0})$ to within additive error $\varepsilon$
with success probability at least $0.9$, using $S(\Theta(1/\varepsilon),0.01)$
samples of $\ket{\psi}$ and $T(\Theta(1/\varepsilon),0.01) + O(n/\varepsilon)$ two-qubit gates,
where the $O(n/\varepsilon)$ term comes from the non-query gate complexity of
the original query algorithm in \cref{lemma:subroutine-td}.
        Note that this algorithm can be used to solve $\textsc{PureTD}\sbra{n, 0, 2\varepsilon}$. 
        By \cref{thm:time-puretd}, we have
        \begin{equation}
            T(\Theta(1/\varepsilon),0.01) + O(n/\varepsilon) \geq \mathsf{T}\rbra{\textsc{PureTD}\sbra{n, 0, 2\varepsilon}} \geq \Omega\rbra{n/\varepsilon^2}. 
        \end{equation}
        By letting $Q = \Theta\rbra{1/\varepsilon}$ due to the arbitrariness of $\varepsilon$, we have
        \begin{equation}
            T\rbra{Q, 0.01} \geq cnQ^2
        \end{equation}
        for sufficiently large $Q > 0$, where $c$ is a universal constant. 
        
        To establish a lower bound on $T\rbra{Q, \delta}$ for arbitrarily small $\delta > 0$, we note that $T\rbra{Q, \delta}$ satisfies the subadditivity that
        \begin{equation}
            T\rbra{Q_1+Q_2, \delta_1+\delta_2} \leq T\rbra{Q_1, \delta_1} + T\rbra{Q_2, \delta_2}
        \end{equation}
        for any integers $Q_1, Q_2 \geq 0$ and real numbers $\delta_1, \delta_2 \in \rbra{0, 1}$. 
        This is because one can always samplize a quantum query algorithm with query complexity $Q_1+Q_2$ to precision $\delta_1 + \delta_2$ by first samplizing the first $Q_1$ queries to precision $\delta_1$ and then samplizing the remaining $Q_2$ queries to precision $\delta_2$. 
        Therefore, we further have 
        \begin{enumerate}
            \item Submultiplicativity: $T\rbra{mQ, m\delta} \leq m \cdot T\rbra{Q, \delta}$ for every integers $m, Q \geq 1$ and real number $\delta > 0$.
            \item Monotonicity: $T\rbra{Q_1, \delta_1} \geq T\rbra{Q_2, \delta_2}$ if $Q_1 \geq Q_2 > 0$ and $0 < \delta_1 \leq \delta_2$. 
        \end{enumerate}
        For $\delta \in \rbra{0, 0.01}$, by taking $m = \floor{\frac{1}{100\delta}} \geq 1$ (which gives $0 < m\delta \leq 0.01$), we have
        \begin{align}
            T\rbra{Q, \delta} 
            & \geq \frac{1}{m} \cdot T\rbra{mQ, m\delta} \\
            & \geq \frac{1}{m} \cdot T\rbra{mQ, 0.01} \\
            & \geq \frac{1}{m} \cdot cn\rbra{mQ}^2 \\
            & = cn \cdot \floor*{\frac{1}{100\delta}} \cdot Q^2 \\
            & \geq \Omega\rbra{nQ^2/\delta}.
        \end{align}
        This means that any implementation of $n$-qubit pure-state samplizer has time complexity $\Omega\rbra{nQ^2/\delta}$. 
    \end{proof}

    \cref{thm:time-samplizer} will be used to prove the time-optimality of the LMR protocol in \cref{sec:lmr}. 

    \subsection{LMR protocol} \label{sec:lmr}

    In this section, we consider the LMR protocol \cite{LMR14,KLL+17,GKP+24}, which is a method to implement the unitary operator $e^{-i\rho t}$ using samples of $\rho$.
    For simplicity, we only consider the LMR protocol for pure states.

    \begin{theorem} [LMR protocol for pure states, adapted from \cite{LMR14,KLL+17,GKP+24}]
        For every $t \in \rbra{0, 2\pi}$ and $n$-qubit pure state $\ket{\psi}$, we can implement the unitary operator $e^{-i\ketbra{\psi}{\psi}t}$ to precision $\delta$ in diamond norm using $O\rbra{1/\delta}$ samples of $\ket{\psi}$ and $O\rbra{n/\delta}$ two-qubit gates. 
    \end{theorem}

    In the following, we show that the LMR protocol is time-optimal for pure states. 
    \begin{theorem} [Quantum time lower bounds for LMR protocol] \label{thm:lmr}
        Any implementation of the unitary operator $e^{-i\ketbra{\psi}{\psi}t}$ to precision $\delta$ in diamond norm using samples of an $n$-qubit pure state $\ket{\psi}$ requires quantum time complexity $\Omega\rbra{n/\delta}$, even if $t = \pi$.
    \end{theorem}
    \begin{proof}
        The case of $t = \pi$ refers to the reflection operator about $\ket{\psi}$: $e^{-i\ketbra{\psi}{\psi}\pi} = I - 2\ketbra{\psi}{\psi} = R_{\psi}$. 
        Suppose that we can implement $R_{\psi}$ to precision $\delta$ in diamond norm using $S\rbra{\delta}$ samples of $\ket{\psi}$ and $T\rbra{\delta}$ two-qubit gates. 
        Then, using the construction of the pure-state samplizer in \cite[Theorem 6.3]{WZ24c}, there is an implementation of pure-state samplizer with sample complexity $Q \cdot S\rbra{\delta/Q}$ and time complexity $Q \cdot T\rbra{\delta/Q}$. 
        By \cref{thm:time-samplizer}, we have 
        \begin{equation}
            Q \cdot T\rbra{\delta/Q} \geq \Omega\rbra{nQ^2/\delta}.
        \end{equation}
        for sufficiently large integer $Q \geq 1$ and every $\delta \in \rbra{0, 0.01}$. 
        Finally, by letting $\delta' = \delta/Q$, we have $T\rbra{\delta'} \geq \Omega\rbra{n/\delta'}$. 
        This means that for sufficiently small $\delta' > 0$, any implementation of $R_{\psi}$ using samples of $\ket{\psi}$ has to use $\Omega\rbra{n/\delta'}$ two-qubit gates. 
    \end{proof}

    \addcontentsline{toc}{section}{References}

    \bibliographystyle{alphaurl}
    \bibliography{main}

@article{Hel67,
    author = {Helstrom, Carl W.},
    title = {Detection theory and quantum mechanics},
    journal = {Information and Control},
    volume = {10},
    number = {3},
    pages = {254--291},
    doi = {10.1016/S0019-9958(67)90302-6},
    year = {1967}
}

@article{Hol73,
    author = {Holevo, Alexander S.},
    title = {Statistical decision theory for quantum systems},
    journal = {Journal of Multivariate Analysis},
    volume = {3},
    number = {4},
    pages = {337--394},
    doi = {10.1016/0047-259X(73)90028-6},
    year = {1973}
}

@article{CWZ24,
    author = {Chen, Kean and Wang, Qisheng and Zhang, Zhicheng},
    title = {Local test for unitarily invariant properties of bipartite quantum states},
    journal = {IEEE Transactions on Information Theory},
    volume = {},
    number = {},
    pages = {},
    doi = {10.1109/TIT.2026.3697790},
    eprint = {2404.04599},
    year = {2026}
}

@inproceedings{SW22,
    title = {Testing matrix product states},
    author = {Soleimanifar, Mehdi and Wright, John},
    booktitle = {Proceedings of the 2022 Annual ACM-SIAM Symposium on Discrete Algorithms},
    pages = {1679--1701},
    doi = {10.1137/1.9781611977073.68},
    year = {2022},
}

@article{BCWdW01,
    author = {Buhrman, Harry and Cleve, Richard and Watrous, John and de Wolf, Ronald},
    title = {Quantum Fingerprinting},
    journal = {Physical Review Letters},
    volume = {87},
    number = {16},
    pages = {167902},
    doi = {10.1103/PhysRevLett.87.167902},
    year = {2001}
}

@article{EAO+02,
    author = {Ekert, Artur K. and Alves, Carolina Moura and Oi, Daniel K. L. and Horodecki, Micha{\l} and Horodecki, Pawe{\l} and Kwek, L. C.},
    title = {Direct estimations of linear and nonlinear functionals of a quantum state},
    journal = {Physical Review Letters},
    volume = {88},
    number = {21},
    pages = {217901},
    doi = {10.1103/PhysRevLett.88.217901},
    year = {2002}
}

@article{HM13,
    author = {Harrow, Aram W. and Montanaro, Ashley},
    title = {Testing product states, quantum {Merlin-Arthur} games and tensor optimization},
    journal = {Journal of the ACM},
    volume = {60},
    number = {1},
    pages = {3:1--3:43},
    doi = {10.1145/2432622.2432625},
    year = {2013}
}

@article{LMR14,
    author = {Lloyd, Seth and Mohseni, Masoud and Rebentrost, Patrick},
    journal = {Nature Physics},
    title = {Quantum principal component analysis},
    volume = {10},
    number = {9},
    pages = {631-633},
    doi = {10.1038/nphys3029},
    year = {2014}
}

@article{KLL+17,
    author = {Kimmel, Shelby and Lin, Cedric Yen-Yu and Low, Guang Hao and Ozols, Maris and Yoder, Theodore J.},
    title = {Hamiltonian simulation with optimal sample complexity},
    journal = {npj Quantum Information},
    volume = {3},
    number = {1},
    pages = {1--7},
    doi = {10.1038/s41534-017-0013-7},
    year = {2017}
}

@article{WZ24,
    author = {Wang, Qisheng and Zhang, Zhicheng},
    title = {Time-efficient quantum entropy estimator via samplizer},
    journal = {IEEE Transactions on Information Theory},
    volume = {71},
    number = {12},
    pages = {9569--9599},
    doi = {10.1109/TIT.2025.3576137},
    year = {2025}
}

@incollection{MdW16,
    author = {Montanaro, Ashley and de Wolf, Ronald},
    title = {A survey of quantum property testing},
    year = {2016},
    publisher = {University of Chicago},
    booktitle = {Theory of Computing Library},
    series = {Graduate Surveys},
    number = {7},
    doi = {10.4086/toc.gs.2016.007},
    pages = {1--81},
}

@article{BV97,
    author = {Bernstein, Ethan and Vazirani, Umesh},
    title = {Quantum complexity theory},
    journal = {SIAM Journal on Computing},
    volume = {26},
    number = {5},
    pages = {1411--1473},
    doi = {10.1137/S0097539796300921},
    year = {1997}
}

@article{Amb07,
    author = {Ambainis, Andris},
    title = {Quantum walk algorithm for element distinctness},
    journal = {SIAM Journal on Computing},
    volume = {37},
    number = {1},
    pages = {210--239},
    doi = {10.1137/S0097539705447311},
    year = {2007}
}

@misc{Jef22,
    author = {Jeffery, Stacey},
    title = {Quantum subroutine composition},
    eprint = {2209.14146},
    howpublished = {ArXiv e-prints},
    year = {2022}
}

@article{BJY24,
    author = {Belovs, Aleksandrs and Jeffery, Stacey and Yolcu, Duyal},
    title = {Taming quantum time complexity},
    journal = {Quantum},
    volume = {8},
    number = {},
    pages = {1444},
    doi = {10.22331/q-2024-08-23-1444},
    year = {2024}
}

@inproceedings{BCJ+13,
    author = {Belovs, Aleksandrs and Childs, Andrew M. and Jeffery, Stacey and Kothari, Robin and Magniez, Frédéric},
    title = {Time-efficient quantum walks for 3-distinctness},
    booktitle = {Proceedings of the 40th International Colloquium on Automata, Languages, and Programming},
    pages = {105--122},
    doi = {10.1007/978-3-642-39206-1_10},
    year = {2013}
}

@inproceedings{JZ23,
    author = {Jeffery, Stacey and Zur, Sebastian},
    title = {Multidimensional quantum walks},
    booktitle = {Proceedings of the 55th Annual ACM Symposium on Theory of Computing},
    pages = {1125--1130},
    doi = {10.1145/3564246.3585158},
    year = {2023}
}

@inproceedings{CJOP20,
    author = {Cornelissen, Arjan and Jeffery, Stacey and Ozols, Maris and Piedrafita, Alvaro},
    title = {Span programs and quantum time complexity},
    booktitle = {Proceedings of the 45th International Symposium on Mathematical Foundations of Computer Science},
    pages = {26:1--26:14},
    doi = {10.4230/LIPIcs.MFCS.2020.26},
    year = {2020}
}

@misc{ABB+23,
    author = {Allcock, Jonathan and Bao, Jinge and Belovs, Aleksandrs and Lee, Troy and Santha, Miklos},
    title = {On the quantum time complexity of divide and conquer},
    howpublished = {ArXiv e-prints},
    eprint = {2311.16401},
    year = {2023}
}

@article{BCH06,
    author = {Bacon, Dave and Chuang, Isaac L. and Harrow, Aram W.},
    title = {Efficient quantum circuits for {Schur} and {Clebsch-Gordan} transforms},
    journal = {Physical Review Letters},
    volume = {97},
    number = {17},
    pages = {170502},
    doi = {10.1103/PhysRevLett.97.170502},
    year = {2006}
}

@inproceedings{BCH07,
    author = {Bacon, Dave and Chuang, Isaac L. and Harrow, Aram W.},
    title = {The quantum {Schur} and {Clebsch-Gordan} transforms: {I.} efficient qudit circuits},
    booktitle = {Proceedings of the 18th Annual ACM-SIAM Symposium on Discrete Algorithms},
    pages = {1235--1244},
    doi = {},
    url = {https://dl.acm.org/doi/abs/10.5555/1283383.1283516},
    year = {2007}
}

@inproceedings{CHW07,
    author = {Childs, Andrew M. and Harrow, Aram W. and Wocjan, Pawe{\l}},
    title = {Weak {Fourier-Schur} sampling, the hidden subgroup problem, and the quantum collision problem},
    booktitle = {Proceedings of the 24th Annual Symposium on Theoretical Aspects of Computer Science},
    pages = {598-609},
    doi = {10.1007/978-3-540-70918-3_51},
    year = {2007}
}

@misc{Ngu23,
    author = {Nguyen, Quynh T.},
    title = {The mixed {Schur} transform: efficient quantum circuit and applications},
    howpublished = {ArXiv e-prints},
    eprint = {2310.01613},
    year = {2023}
}

@misc{GBO23,
    author = {Grinko, Dmitry and Burchardt, Adam and Ozols, Maris},
    title = {{Gelfand-Tsetlin} basis for partially transposed permutations, with applications to quantum information},
    howpublished = {ArXiv e-prints},
    eprint = {2310.02252},
    year = {2023}
}

@misc{GKP+24,
    author = {Go, Byeongseon and Kwon, Hyukjoon and Park, Siheon and Patel, Dhrumil and Wilde, Mark M.},
    title = {Density matrix exponentiation and sample-based {Hamiltonian} simulation: Non-asymptotic analysis of sample complexity},
    howpublished = {ArXiv e-prints},
    eprint = {2412.02134},
    year = {2024}
}

@article{HHJ+17,
    author = {Haah, Jeongwan and Harrow, Aram W. and Ji, Zhengfeng and Wu, Xiaodi and Yu, Nengkun},
    title = {Sample-optimal tomography of quantum states},
    journal = {IEEE Transactions on Information Theory},
    volume = {63},
    number = {9},
    pages = {5628--5641},
    doi = {10.1109/TIT.2017.2719044},
    year = {2017}
}

@inproceedings{OW16,
    author = {O'Donnell, Ryan and Wright, John},
    title = {Efficient quantum tomography},
    booktitle = {Proceedings of the 48th Annual ACM Symposium on Theory of Computing},
    pages = {899--912},
    doi = {10.1145/2897518.2897544},
    year = {2016}
}

@inproceedings{OW17,
    author = {O'Donnell, Ryan and Wright, John},
    title = {Efficient quantum tomography {II}},
    booktitle = {Proceedings of the 49th Annual ACM Symposium on Theory of Computing},
    pages = {962--974},
    doi = {10.1145/3055399.3055454},
    year = {2017}
}

@inproceedings{BOW19,
    author = {B{\u{a}}descu, Costin and O'Donnell, Ryan and Wright, John},
    title = {Quantum state certification},
    booktitle = {Proceedings of the 51st Annual ACM SIGACT Symposium on Theory of Computing},
    pages = {503--514},
    doi = {10.1145/3313276.3316344},
    year = {2019}
}

@article{OW21,
    author = {O'Donnell, Ryan and Wright, John},
    title = {Quantum spectrum testing},
    journal = {Communications in Mathematical Physics},
    volume = {387},
    number = {1},
    pages = {1--75},
    doi = {10.1007/s00220-021-04180-1},
    year = {2021}
}

@article{WZ24b,
    author = {Wang, Qisheng and Zhang, Zhicheng},
    title = {Fast quantum algorithms for trace distance estimation},
    journal = {IEEE Transactions on Information Theory},
    volume = {70},
    number = {4},
    pages = {2720--2733},
    doi = {10.1109/TIT.2023.3321121},
    year = {2024}
}

@article{AISW20,
    author = {Acharya, Jayadev and Issa, Ibrahim and Shende, Nirmal V. and Wagner, Aaron B.},
    title = {Estimating Quantum Entropy},
    journal = {IEEE Journal on Selected Areas in Information Theory},
    volume = {1},
    number = {2},
    pages = {454--468},
    doi = {10.1109/JSAIT.2020.3015235},
    year = {2020}
}

@misc{BMW16,
    author = {Bavarian, Mohammad and Mehraban, Saeed and Wright, John},
    title = {Learning entropy},
    howpublished = {A manuscript on von Neumann entropy estimation, private communication},
    year = {2016}
}

@article{BBC+01,
    author = {Beals, Robert and Buhrman, Harry and Cleve, Richard and Mosca, Michele and de Wolf, Ronald},
    title = {Quantum lower bounds by polynomials},
    journal = {Journal of the ACM},
    volume = {48},
    number = {4},
    pages = {778--797},
    doi = {10.1145/502090.502097},
    year = {2001}
}

@article{Amb02,
    author = {Ambainis, Andris},
    title = {Quantum lower bounds by quantum arguments},
    journal = {Journal of Computer and System Sciences},
    volume = {64},
    number = {4},
    pages = {750--767},
    doi = {10.1006/jcss.2002.1826},
    year = {2002}
}

@inproceedings{Gro96,
    author = {Grover, Lov K.},
    title = {A fast quantum mechanical algorithm for database search},
    booktitle = {Proceedings of the 28th Annual ACM Symposium on Theory of Computing},
    pages = {212--219},
    doi = {10.1145/237814.237866},
    year = {1996}
}

@article{BBBV97,
    author = {Bennett, Charles H. and Bernstein, Ethan and Brassard, Gilles and Vazirani, Umesh},
    title = {Strengths and weaknesses of quantum computing},
    journal = {SIAM Journal on Computing},
    volume = {26},
    number = {5},
    pages = {1510--1523},
    doi = {10.1137/S0097539796300933},
    year = {1997}
}

@article{BBHT98,
    author = {Boyer, Michel and Brassard, Gilles and H{\o}yer, Peter and Tapp, Alain},
    title = {Tight bounds on quantum searching},
    journal = {Fortschritte der Physik},
    volume = {46},
    number = {4-5},
    pages = {493-505},
    doi = {10.1002/(SICI)1521-3978(199806)46:4/5<493::AID-PROP493>3.0.CO;2-P},
    year = {1998}
}

@article{Zal99,
    author = {Zalka, Christof},
    title = {Grover's quantum searching algorithm is optimal},
    journal = {Physical Review A},
    volume = {60},
    number = {4},
    pages = {2746},
    doi = {10.1103/PhysRevA.60.2746},
    year = {1999}
}

@article{Gro02,
    author = {Grover, Lov K.},
    title = {Trade-offs in the quantum search algorithm},
    journal = {Physical Review A},
    volume = {66},
    number = {5},
    pages = {052314},
    doi = {10.1103/PhysRevA.66.052314},
    year = {2002}
}

@article{AdW17,
    author = {Arunachalam, Srinivasan and de Wolf, Ronald},
    title = {Optimizing the number of gates in quantum search},
    journal = {Quantum Information and Computation},
    volume = {17},
    number = {3--4},
    pages = {251--261},
    doi = {10.26421/QIC17.3-4-4},
    year = {2017}
}

@article{CWLY23,
    title={Unitarity estimation for quantum channels},
    author={Chen, Kean and Wang, Qisheng and Long, Peixun and Ying, Mingsheng},
    journal={IEEE Transactions on Information Theory},
    volume={69},
    number={8},
    pages={303--326},
    doi={10.1109/TIT.2023.3263645},
    year={2023}
}

@misc{GHYZ24,
    author = {Gong, Weiyuan and Haferkamp, Jonas and Ye, Qi and Zhang, Zhihan},
    title = {On the sample complexity of purity and inner product estimation},
    eprint = {2410.12712},
    howpublished = {ArXiv e-prints},
    year = {2024}
}

@inproceedings{ALL22,
    author = {Anshu, Anurag and Landau, Zeph and Liu, Yunchao},
    title = {Distributed quantum inner product estimation},
    booktitle = {Proceedings of the 54th Annual ACM SIGACT Symposium on Theory of Computing},
    pages = {44--51},
    doi = {10.1145/3519935.3519974},
    year = {2022}
}

@article{CWYZ25,
    author = {Chen, Kean and Wang, Qisheng and Yu, Zhan and Zhang, Zhicheng},
    title = {Simultaneous estimation of nonlinear functionals of a quantum state},
    journal = {IEEE Transactions on Information Theory},
    volume = {},
    number = {},
    pages = {},
    doi = {10.1109/TIT.2026.3699531},
    year = {2026}
}

@article{MKB05,
    author = {Mintert, Florian and Kuś, Marek and Buchleitner, Andreas},
    title = {Concurrence of mixed multipartite quantum states},
    journal = {Physical Review Letters},
    volume = {95},
    number = {26},
    pages = {260502},
    doi = {10.1103/PhysRevLett.95.260502},
    year = {2005}
}

@inproceedings{WZ24c,
    author = {Wang, Qisheng and Zhang, Zhicheng},
    title = {Sample-optimal quantum estimators for pure-state trace distance and fidelity via samplizer},
    booktitle = {Proceedings of the 53rd International Colloquium on Automata, Languages, and Programming},
    pages = {},
    doi = {},
    eprint = {2410.21201},
    year = {2026}
}

@article{Wan24,
    title={Optimal trace distance and fidelity estimations for pure quantum states},
    author={Wang, Qisheng},
    journal={IEEE Transactions on Information Theory},
    volume={70},
    number={12},
    pages={8791--8805},
    year={2024},
    doi={10.1109/TIT.2024.3447915}
}

@inbook{Wat18,
    author = {Watrous, John},
    title = {The Theory of Quantum Information},
    chapter = {7 - Permutation Invariance and Unitarily Invariant Measures},
    publisher = {Cambridge University Press},
    year = {2018},
    doi = {10.1017/9781316848142.008},
}

@inproceedings{ACL+20,
    author = {Aaronson, Scott and Chia, Nai-Hui and Lin, Han-Hsuan and Wang, Chunhao and Zhang, Ruizhe},
    title = {On the quantum complexity of closest pair and related problems},
    booktitle = {Proceedings of the 35th Computational Complexity Conference},
    pages = {16:1--16:43},
    doi = {10.4230/LIPIcs.CCC.2020.16},
    year = {2020}
}

@inproceedings{BPS21,
    author = {Buhrman, Harry and Patro, Subhasree and Speelman, Florian},
    title = {A framework of quantum strong exponential-time hypotheses},
    booktitle = {Proceedings of the 38th International Symposium on Theoretical Aspects of Computer Science},
    pages = {19:1--19:19},
    doi = {10.4230/LIPIcs.STACS.2021.19},
    year = {2021}
}

@incollection{FR09,
    author = {French, Steven and Rickles, Dean},
    title = {Understanding permutation symmetry},
    editor = {Brading, Katherine and Castellani, Elena},
    booktitle = {Symmetries in Physics: Philosophical Reflections},
    volume = {},
    number = {},
    pages = {212--238},
    doi = {10.1017/CBO9780511535369.013},
    publisher = {Cambridge University Press},
    series = {},
    year = {2009}
}

@article{TWG+10,
    author = {Tóth, G. and Wieczorek, W. and Gross, D. and Krischek, R. and Schwemmer, C. and Weinfurter, H.},
    title = {Permutationally invariant quantum tomography},
    journal = {Physical Review Letters},
    volume = {105},
    number = {25},
    pages = {250403},
    doi = {10.1103/PhysRevLett.105.250403},
    year = {2010}
}

@article{PR04,
    author = {Pollatsek, Harriet and Ruskai, Mary Beth},
    title = {Permutationally invariant codes for quantum error correction},
    journal = {Linear Algebra and its Applications},
    volume = {392},
    number = {},
    pages = {255--288},
    doi = {10.1016/j.laa.2004.06.014},
    year = {2004}
}

@article{Ouy14,
    author = {Ouyang, Yingkai},
    title = {Permutation-invariant quantum codes},
    journal = {Physical Review A},
    volume = {90},
    number = {6},
    pages = {062317},
    doi = {10.1103/PhysRevA.90.062317},
    year = {2014}
}

@article{AA14,
    author = {Aaronson, Scott and Ambainis, Andris},
    title = {The need for structure in quantum speedups},
    journal = {Theory of Computing},
    volume = {10},
    number = {6},
    pages = {133--166},
    doi = {10.4086/toc.2014.v010a006},
    year = {2014}
}

@inproceedings{Cha19,
    author = {Chailloux, André},
    title = {A note on the quantum query complexity of permutation symmetric functions},
    booktitle = {Proceedings of the 10th Innovations in Theoretical Computer Science Conference},
    pages = {19:1--19:7},
    doi = {10.4230/LIPIcs.ITCS.2019.19},
    year = {2019}
}

@article{BDCG+24,
    author = {Ben-David, Shalev and Childs, Andrew M. and Gilyén, András and Kretschmer, William and Podder, Supartha and Wang, Daochen},
    title = {Symmetries, graph properties, and quantum speedups},
    journal = {SIAM Journal on Computing},
    volume = {53},
    number = {6},
    pages = {FOCS20-368--FOCS20-415},
    doi = {10.1137/23M1573975},
    year = {2024}
}

@article{GHYY25,
    author = {Guan, Ziyi and Huang, Yunqi and Yao, Penghui and Ye, Zekun},
    title = {Quantum and classical communication complexity of permutation-invariant functions},
    journal = {IEEE Transactions on Information Theory},
    volume = {},
    number = {},
    pages = {},
    doi = {10.1109/TIT.2025.3534920},
    year = {2025}
}

@article{LW25,
    title={On estimating the trace of quantum state powers},
    author={Liu, Yupan and Wang, Qisheng},
    journal={IEEE Transactions on Information Theory},
    volume={},
    number={},
    pages={},
    year={2026},
    doi={10.1109/TIT.2026.3683891}
}

@article{WGL+24,
    title = {New quantum algorithms for computing quantum entropies and distances},
    author = {Wang, Qisheng and Guan, Ji and Liu, Junyi and Zhang, Zhicheng and Ying, Mingsheng},
    journal = {IEEE Transactions on Information Theory},
    volume = {70},
    number = {8},
    pages = {5653--5680},
    year = {2024},
    doi = {10.1109/TIT.2024.3399014},
}

@article{SH21,
    author = {Subramanian, Sathyawageeswar and Hsieh, Min-Hsiu},
    title = {Quantum algorithm for estimating $\alpha$-Renyi entropies of quantum states},
    journal = {Physical Review A},
    volume = {104},
    number = {2},
    pages = {022428},
    doi = {10.1103/PhysRevA.104.022428},
    year = {2021}
}

@article{WZL24,
    author = {Wang, Xinzhao and Zhang, Shengyu and Li, Tongyang},
    title = {A quantum algorithm framework for discrete probability distributions with applications to {R\'{e}nyi} entropy estimation},
    journal = {IEEE Transactions on Information Theory},
    volume = {70},
    number = {5},
    pages = {3399--3426},
    doi = {10.1109/TIT.2024.3382037},
    year = {2024}
}

@inproceedings{CW25,
    author = {Chen, Kean and Wang, Qisheng},
    title = {Improved sample upper and lower bounds for trace estimation of quantum state powers},
    booktitle = {Proceedings of the 38th Annual Conference on Learning Theory},
    pages = {1008-1028},
    doi = {},
    url = {https://proceedings.mlr.press/v291/chen25d.html},
    year = {2025}
}

@misc{CLW26,
    author = {Chen, Kean and Liu, Yupan and Wang, Qisheng},
    title = {Trace estimation of quantum state powers: Sample complexity and computational hardness},
    howpublished = {ArXiv e-prints},
    eprint = {2505.09563v2},
    year = {2026}
}

@article{WZ25,
    author = {Wang, Qisheng and Zhang, Zhicheng},
    title = {Quantum lower bounds by sample-to-query lifting},
    journal = {SIAM Journal on Computing},
    volume = {54},
    number = {5},
    pages = {1294--1334},
    doi = {10.1137/24M1638616},
    year = {2025}
}

@article{Bra03,
    author = {Brassard, Gilles},
    title = {Quantum communication complexity},
    journal = {Foundations of Physics},
    volume = {33},
    number = {11},
    pages = {1593--1616},
    doi = {10.1023/A:1026009100467},
    year = {2003}
}

@inproceedings{MdW23,
    title = {Tight bounds for quantum phase estimation and related problems},
    author = {Mande, Nikhil S. and de Wolf, Ronald},
    booktitle = {Proceedings of the 31st Annual European Symposium on Algorithms},
    pages = {81:1--81:16},
    doi = {10.4230/LIPIcs.ESA.2023.81},
    year = {2023},
}

@book{NC10,
    author = {Nielsen, Michael A. and Chuang, Isaac L.},
    title = {Quantum Computation and Quantum Information},
    publisher = {Cambridge University Press},
    doi = {10.1017/CBO9780511976667},
    year = {2010}
}

    \appendix

    \section{An Example with Low Time-to-Sample Ratio} \label{sec:low-ratio}

    In this appendix, we provide an $n$-qubit property $\mathcal{P}_n$ such that the time-to-sample ratio $R = \mathsf{T}\rbra{\mathcal{P}_n}/\mathsf{S}\rbra{\mathcal{P}_n} = \polylog\rbra{n}$ is low, where every qubit contains useful information. 
    
    Without loss of generality, let $n = 2^{t}-1$ where $t \geq 1$. 
    Let $S$ denote the set of the $n$-qubit pure states $\ket{\psi}$ with the following property:
    \begin{enumerate}
        \item $\ket{\psi} = \ket{x_1} \ket{x_2} \dots \ket{x_n}$, where $x_j \in \cbra{0, 1}$ for $j \in \sbra{n}$. 
        \item $x_1 + x_2 + \dots + x_n = t$. 
        \item Let $p_1 < p_2 < \dots < p_t$ be the positions with $x_{p_j} = 1$ for $j \in \sbra{t}$. 
        Then, $\floor{p_{j+1}/2} = p_j$ for every $1 \leq j < t$.
    \end{enumerate}
    Intuitively, the sequence $x_1, x_2, \dots, x_n$ encodes a root-to-leaf path in a perfect binary tree of $n$ nodes, where node $1$ is the root and nodes $2^{t-1}, \dots, 2^t-1$ are leaf nodes. 
    For $1 \leq j < 2^{t-1}$, node $j$ has two child nodes $2j$ and $2j+1$.

    Now we define an $n$-qubit property $\mathcal{P}_n$ as follows: 
    \begin{align}
        \mathcal{P}_n^{\textup{yes}} & = \set{\ket{\psi} \in S}{\text{the path encoded by $\ket{\psi}$ leads to an even-numbered leaf node}}, \\
        \mathcal{P}_n^{\textup{no}} & = \set{\ket{\psi} \in S}{\text{the path encoded by $\ket{\psi}$ leads to an odd-numbered leaf node}}.
    \end{align}
    It can be seen that $\mathsf{S}\rbra{\mathcal{P}_n} = 1$. 
    This is because for every $\ket{\psi} \in S$, we can measure all the $n$ qubits of $\ket{\psi}$ in the computational basis and then determine the parity of the leaf node to which the path encoded by $\ket{\psi}$ leads. 
    However, this simple approach only gives a time complexity of $\mathsf{T}\rbra{\mathcal{P}_n} = O\rbra{n}$. 
    
    In the following, we show that $\mathsf{T}\rbra{\mathcal{P}_n} = O\rbra{\log^2\rbra{n}}$. 
    For a state $\ket{\psi} \in S$, we initiate a variable $j \gets 1$ (the root).
    As long as node $j$ is not a leaf node, we check if the $2j$-th qubit of $\ket{\psi}$ contains $1$ (this costs $O\rbra{\log\rbra{j}}$ time for the indirect addressing by the variable $j$): if it contains $1$, set $j \gets 2j$; otherwise, set $j \gets 2j+1$. 
    When node $j$ is a leaf node, we can determine whether $\ket{\psi} \in \mathcal{P}_n^{\textup{yes}}$ or $\ket{\psi} \in \mathcal{P}_n^{\textup{no}}$ by checking the parity of $j$. 
    Therefore, the overall time complexity is 
    \begin{equation}
    \mathsf{T}\rbra{\mathcal{P}_n} = \sum_{k=1}^t O\rbra{k} = O\rbra{t^2} = O\rbra{\log^2\rbra{n}}.
    \end{equation}
\end{document}